\documentclass[11pt]{article}
\usepackage[T1]{fontenc}
\usepackage{geometry} 

\usepackage[english]{babel}
 
\usepackage{comment}

\usepackage{amsmath,amsthm}

\usepackage{graphicx}

\usepackage{hyperref}

\makeatletter
\newtheorem*{rep@theorem}{\rep@title}
\newcommand{\newreptheorem}[2]{%
	\newenvironment{rep#1}[1]{%
		\def\rep@title{#2 \ref{##1}}%
		\begin{rep@theorem}}%
		{\end{rep@theorem}}}
\makeatother

\newenvironment{lemma-repeat}[1]{\begin{trivlist}
		\item[\hspace{\labelsep}{\bf\noindent Lemma \ref{#1} }]\em }%
	{\end{trivlist}}
\newenvironment{theorem-repeat}[1]{\begin{trivlist}
		\item[\hspace{\labelsep}{\bf\noindent Theorem \ref{#1} }]\em }%
	{\end{trivlist}}

\newcommand{\qedsymb}{\qed}

\newtheorem{theorem}{Theorem}[section]
\newtheorem{claim}[theorem]{Claim}
\newtheorem{lemma}[theorem]{Lemma}

\newcounter{challenge}[section]

\DeclareMathOperator{\poly}{poly}

\newcommand{\remove}[1]{}

\usepackage{xcolor}

\usepackage{amsmath}
\usepackage{amsfonts}
\usepackage[linesnumbered,vlined]{algorithm2e}

\usepackage{comment}
\usepackage{amsthm}
\usepackage[inline]{enumitem}
\graphicspath{ {./images/} }
 
\usepackage{latexsym}

\usepackage{todonotes}
\usepackage{comment}
\usepackage{mathtools}
\usepackage{tikz}

\theoremstyle{remark}

\theoremstyle{definition}

\newcommand{\Sstar}{S^\ast}
\newcommand{\Cset}{\mathbb{C}}
\newcommand{\Cstarset}{\mathbb{C}^\ast}

\newcommand{\E}{\mathbf{E}}

\newcommand{\vol}{\textup{vol}}

\DeclarePairedDelimiter{\floor}{\lfloor}{\rfloor}

\newcommand{\ip}[1]{\left}
\usepackage{xspace}

\newcommand{\myn}{\bar{n}}

\newcommand{\congest}{\ensuremath{\mathsf{CONGEST}}\xspace}

\newcommand{\clique}{\ensuremath{\mathsf{CONGESTED~CLIQUE}}\xspace}

\newcommand{\mix}{\ensuremath{\tau_{\operatorname{mix}}}}

\newcommand{\mainTheoremText}{For each $p\ge 4$, the $K_p$-listing problem can be solved in $\tilde{O}(n^{1 - 2/p})$ rounds with high probability in the \congest model.}
\newcommand{\polylog}{\operatorname{polylog}}
\newcommand{\ID}{\operatorname{ID}}

\usepackage{cite}
\newcommand{\avgdeg}{\mu}

	\title{
		Tight Distributed Listing of Cliques
	}

\author{
Keren Censor-Hillel \\
\small Technion \\
\and
Yi-Jun Chang\\
\small ETH Z\"{u}rich \\
\and
Fran\c{c}ois Le Gall\\
\small Nagoya University \\
\and
Dean Leitersdorf\\
\small Technion  
 }

\begin{document}
\date{}
\maketitle
\thispagestyle{empty}
\setcounter{page}{0}

	\begin{abstract}
Much progress has recently been made in understanding the complexity landscape of subgraph finding problems in the \congest model of distributed computing. However, so far, very few tight bounds are known in this area. For triangle (i.e., 3-clique) listing, an optimal $\tilde{O}(n^{1/3})$-round distributed algorithm has been constructed by Chang et al.~[SODA 2019, PODC 2019]. Recent works of Eden et al.~[DISC 2019] and of Censor-Hillel et al.~[PODC 2020] have shown sublinear algorithms for $K_p$-listing, for each $p \geq 4$, but still leaving a significant gap between the upper bounds and the known lower bounds of the problem.

In this paper, we completely close this gap. We show that for each $p \geq 4$, there is an $\tilde{O}(n^{1 - 2/p})$-round distributed algorithm that lists all $p$-cliques $K_p$ in the communication network. Our algorithm is \emph{optimal} up to a polylogarithmic factor, due to the $\tilde{\Omega}(n^{1 - 2/p})$-round lower bound of Fischer et al.~[SPAA 2018], which holds even in the \clique model. Together with the triangle-listing algorithm by Chang et al.~[SODA 2019, PODC 2019], our result thus shows that the  round complexity of $K_p$-listing, for all $p$, is the same in both the \congest and \clique models, at $\tilde{\Theta}(n^{1 - 2/p})$ rounds.

For $p=4$, our result additionally matches the $\tilde{\Omega}(n^{1/2})$ lower bound for $K_4$-\emph{detection} by Czumaj and Konrad [DISC 2018], implying that the round complexities for detection and listing of $K_4$ are equivalent in the \congest model.

	\end{abstract}
	
	\newpage

\section{Introduction}
\emph{Subgraph detection} and \emph{listing} are fundamental graph problems that have been extensively studied in various computational models~\cite{Alon+SIDMA08,becchetti+KDD08,Eden+SICOMP17,hu+JCSS16,LeGall+FOCS14,Shun+IDCE15}. In this paper, we focus on the \congest model of distributed computing, where the communication network is identical to the $n$-node input graph $G=(V,E)$. In this model, each node $v \in V$ represents a computing device, and each edge $e \in E$ represents a communication link.  Each node $v \in V$ has an $O(\log n)$-bit unique identifier $\ID(v)$. The communication proceeds in synchronous rounds. In each round, each node $v$ can send an $O(\log n)$-bit message along each edge $e$ incident to $v$. 

Given a fixed graph $H$, the  \emph{$H$-detection} problem requires that at least one node in the network detects a copy of $H$ if the underlying network $G$ contains $H$ as a subgraph, and the  \emph{$H$-listing} problem requires that each subgraph $H$ of $G$ is detected by some node in the network.

\subsection{Prior Work on Distributed Clique Listing}
The $K_p$-listing problem, for all $p$, can be solved trivially in $\tilde{O}(\Delta)$ rounds by having each node $v$ broadcast the list $\{ \ID(u) \ | \ u \in N(v) \}$ to all its neighbors $N(v)$, where $\Delta < n$ is the maximum degree of the graph. The first breakthrough in this area
is by
Izumi and Le Gall~\cite{Izumi+PODC17}, who showed that $K_3$-detection and listing can be solved in $\tilde{O}(n^{2/3})$ rounds and $\tilde{O}(n^{3/4})$ rounds, respectively. 

Later, Chang et al.~\cite{Chang+SODA19} and Chang and Saranurak~\cite{Chang+PODC19}  brought the round complexity down to $\tilde{O}(n^{1/3})$, matching the $\tilde{\Omega}(n^{1/3})$ lower bound~\cite{Izumi+PODC17,Pandurangan+SPAA18} by a polylogarithmic factor. The main idea underlying the approach of~\cite{Chang+SODA19,Chang+PODC19} is \emph{expander decompositions}. An expander decomposition removes $\epsilon$-fraction of the edges so that
the remaining connected components have conductance at least $\phi$. It was shown in~\cite{Chang+SODA19,Chang+PODC19} that an expander decomposition with  parameters $\epsilon = 1 / \polylog (n)$ and $\phi = 1 / \polylog (n)$ can be constructed in $O(n^{0.001})$ rounds.  

Once an expander decomposition is constructed, we can apply \emph{expander routing}~\cite{Ghaffari+PODC17,Ghaffari+DISC18} to each high-conductance cluster $C$ in the expander decomposition. Specifically, after $O(n^{0.001})$ rounds of pre-processing, within $\poly(\phi^{-1}, \log n)$ rounds we can let each $v \in C$  communicate with any \emph{arbitrary} $\deg_{C}(v)$ nodes in $C$, not just the local neighbors of $v$~\cite{Chang+PODC19}, where $\deg_{C}(v)$ is the number of neighbors of $v$ in $C$. 

Using expander routing, $K_p$-listing can be solved in $\tilde{O}(n^{1 - 2/p})$ rounds on graphs with conductance $\phi = 1/\polylog (n)$~\cite{Chang+SODA19}. Based on this result, the $K_3$-listing algorithm of~\cite{Ghaffari+PODC17,Ghaffari+DISC18} works as follows. Construct an expander decomposition to partition the nodes into high-conductance clusters. For each cluster $C$, use expander routing to list all $K_3$ involving at least one edge in $C$, in parallel, in $\tilde{O}(n^{1/3})$ rounds.
Then recurse on the subgraph induced by the remaining inter-cluster edges.  

Things become complicated when $p \geq 4$, as $K_p$ can involve edges in more than one cluster in this case.
The first sublinear $K_p$-listing algorithms for $p = 4$ and $p = 5$ were given by  Eden et al.~\cite{Eden+DISC19}.
To deal with the cross-cluster clique instances, they classify the nodes outside a cluster $C$ into \emph{heavy} nodes and \emph{light} nodes based on the number of neighbors in $C$. Each heavy node $v$ has sufficiently many neighbors in $C$ so that $v$ has enough bandwidth to send its entire list of neighbors to $C$ efficiently. For each light node $v$, it only needs to send its list of neighbors restricting to those in $C$, and this can be done efficiently since a light node only has a small number of neighbors in~$C$. Choosing the threshold of the classification properly, this information gathering can be done in sublinear rounds. After this step, each cluster $C$ contains all the edges that can potentially form a clique $K_p$ with existing edges in $C$. Finally,  they apply the expander routing to each cluster $C$ to solve the $K_p$-listing problem.
Using this approach, they showed that $K_4$ and $K_5$ can be listed in $\tilde{O}(n^{5/6})$ and $\tilde{O}(n^{21/22})$ rounds, respectively.

Recently, Censor-Hillel et al.~\cite{Censor+PODC20} showed that  $K_p$-listing can be solved in sublinear rounds for \emph{all}~$p$. 
One major shortcoming of the approach of~\cite{Eden+DISC19} is that the number of edges sent to a cluster~$C$ can be significantly larger than the number of edges in $C$. As the total bandwidth in expander routing depends on the number of edges in $C$, this makes the listing algorithm inefficient. To overcome this issue, Censor-Hillel et al.~\cite{Censor+PODC20} proposed the following \emph{arboricity decreasing} framework, which is based on a variant of the expander decomposition considered in~\cite{Chang+SODA19} that allows an additional small arboricity part $E_s$.

The expander decomposition of~\cite{Chang+SODA19} partitions the edge set $E$ into three parts $E_m$, $E_s$, and $E_r$. The set $E_m$ represents the edges inside a high-conductance cluster. The set $E_s$ induces a subgraph of arboricity at most $n^{\delta}$, where $0 < \delta < 1$ is a given parameter. The set $E_r$ is the   remaining edges, and it satisfies $|E_r| < |E|/6$. In this decomposition, each cluster $C$ of $E_m$ not only has conductance $1 / \polylog (n)$ but also each node $v$ in $C$ has $\Omega(n^{\delta})$ neighbors in $C$. Such a decomposition can be constructed in $\tilde{O}(n^{1 - \delta})$ rounds~\cite{Chang+SODA19}.

The idea of~\cite{Censor+PODC20} is to first recurse on $E_r$. After $E_r$ is empty, recurse on $E_s$ with a slightly smaller parameter $\delta$, and so the arboricity of the graph is gradually decreasing during the process. This approach ensures that during the process, the minimum degree in a cluster under consideration is within a small factor to the overall arboricity of the graph.
They showed that $K_p$-listing can be solved in $\tilde{O}(n^{2/3})$ rounds for $p = 4$ and $\tilde{O}(n^{p/(p+2)})$ rounds for $p \geq 5$. 

This still leaves a significant gap between the upper bounds of~\cite{Censor+PODC20} and the   $\tilde{\Omega}(n^{1 - 2/p})$-round lower bound of Fischer et al.~\cite{Fischer+SPAA18}.

\subsection{New Result}

In this paper, we prove the following theorem, which completely closes this gap.

\begin{theorem}
	\label{theorem:Kp}
	\mainTheoremText
\end{theorem}

More precisely, our result matches the $\tilde{\Omega}(n^{1 - 2/p})$ lower bound for $K_p$-listing of Fischer et al.~\cite{Fischer+SPAA18} up to a polylogarithmic factor, and it also matches the $\tilde{\Omega}(n^{1/2})$ lower bound for $K_p$-detection of Czumaj and Konrad~\cite{Czumaj+DISC18} for the case of $p = 4$.

Prior to this work, there were only two  known non-trivial tight bounds in the area of distributed subgraph \emph{listing}: $\tilde{\Theta}(n^{1/3})$ for triangles $K_3$~\cite{Chang+PODC19} and $\tilde{\Theta}(n)$ for 4-cycles $C_4$~\cite{Eden+DISC19}.

The \clique model is a variant of \congest that allows all-to-all communication in each round.
The lower bound of~\cite{Fischer+SPAA18} also holds in the \clique model, and so our result shows that the round complexity of $K_p$-listing is the same in both \congest and \clique up to a polylogarithmic factor, and 
this 
implies that allowing distant nodes to communicate directly does not allow us to list cliques much faster. 
This statement is not true for many other subgraphs. For  $C_4$, detection and listing can be solved in $O(1)$ and $\tilde{O}(n^{1/2})$ rounds in \clique, respectively~\cite{Censor-Hillel+DC19}, but these problems have much higher lower bounds in \congest: $\tilde{\Omega}(n^{1/2})$ for detection~\cite{Drucker+PODC14,Korhonen+OPODIS17} and $\tilde{\Omega}(n)$ for listing~\cite{Eden+DISC19}.

\subsection{Technical Overview}
As discussed earlier, there are two main challenges in dealing with cross-cluster cliques for the case of $p \geq 4$.
The first challenge is that  each cluster $C$ needs to efficiently list all $K_p$ with at least one edge inside $C$. The main difficulty here is that the  number of \emph{relevant} edges outside of $C$ that can form a $K_p$ with existing edges in $C$ can be much higher than the number of edges in $C$. In other words, the problem size for subgraph listing can be much higher than the number of edges that can be used in communication.
The second challenge is that each cluster $C$ needs to efficiently gather all the needed information from outside of  $C$ into $C$. Specifically, we need to let $C$ learn each edge outside of $C$ that can potentially form a $K_p$ instance with edges in $C$.

\paragraph{Optimal sparsity-aware listing.} 
Our main technical ingredient to deal with the above first challenge is an \emph{optimal} sparsity-aware listing algorithm (Theorem~\ref{th:spars} in Section~\ref{section:partition})
We show that the task of listing all instances of $K_p$ with at least one edge in $C$ can still be solved {optimally} in $\tilde{O}(n^{1 - 2/p})$ rounds, as long as the \emph{average degree} in $C=(V_C, E_C)$ satisfies some mild requirements. 
This result works even in the case the number of edges outside of $C$ is \emph{much higher} than the number of edges in $C$.
For comparison, the previous attempts in coping with this challenge mostly only focus on limiting the number of edges outside of $C$ to consider~\cite{Censor+PODC20,Eden+DISC19}.

Before explaining our algorithm, let us review
 the $K_p$-listing algorithm on a high-conductance graph $G=(V,E)$ in~\cite{Chang+SODA19}. Partition the node set $V$ into $n^{1/p}$ parts, so that the number of $p$-tuples of parts is $n$. Associate each node $v \in V$ with $O(\deg(v) \cdot |V| / |E|)$ $p$-tuples, and each node $v$ is responsible for listing all $K_p$ corresponding to the $p$-tuples assigned to $v$. A simple calculation shows that if the partition is done randomly, then each node needs to gather $\deg(v) \cdot O(n^{1 - 2/p})$ edges \emph{in expectation} to fulfill its clique listing task. A more complicated calculation shows that the actual number of edges is concentrated around its expectation, and so we can use the expander routing algorithm of~\cite{Ghaffari+DISC18,Ghaffari+PODC17} to solve the $K_p$-listing problem in $\tilde{O}(n^{1 - 2/p})$ rounds.
Note that this algorithm does not give us the optimal round complexity in our setting, as  the number of edges  under consideration for $K_p$-listing can be much higher than the number of edges that can be used in communication.

The new idea in this paper is to exploit the varying sparsity of different edge sets.
Instead of partitioning the nodes into $n^{1/p}$ parts with roughly equal size, we allow  different parameters for the number of parts for $V_C$ and $V \setminus V_C$ in the random partition.
We show that if the average degree $|E_C| / |V_C|$ of $C$ is sufficiently high, and if each node $v \in V_C$ has $1/\polylog (n)$ fraction of its neighbors in $C$, then 
we can still achieve the optimal round complexity   $\tilde{O}(n^{1 - 2/p})$.

For this to work, we need to show good bounds on the number of edges that we can have between any two  parts in the random partition. To this end, we prove a \emph{partition lemma} (Lemma~\ref{lemma:partition} in Section~\ref{section:partition}), which generalizes~\cite[Lemma 4.2]{Chang+SODA19} to multiple edge sets.

\paragraph{Efficient transmission of edges across clusters.} For the above second challenge, in this paper we adapt an approach similar to the classification of nodes into heavy ones and light ones in~\cite{Eden+DISC19}.
Considering a cluster $C$, we want to gather all the edges that can potentially form a $K_p$ with edges in $C$.
Define $S^\ast$ as the set of nodes $v \in V \setminus V_C$ such that   $\deg_{V \setminus V_C}(v) = \Omega(n^{1 - 2/p} \cdot \deg_C(v))$, i.e., the number of neighbors of $v$ outside $C$ is at least $\Omega(n^{1 - 2/p})$ times the number of neighbors of $v$ inside~$C$. We call the nodes in $S^\ast$ the \emph{light} nodes. Observe that each \emph{heavy} node $v \in (V \setminus V_C) \setminus S^\ast$ can send all its incident edges to $C$ in ${O}(n^{1 - 2/p})$ rounds, due to its high number of neighbors in $C$.

We focus on the light nodes in the subsequent discussion.
Define $S$ as the subset of $V_C$ such that $v \in S$ if $v$ has $\Omega(n^{1 - 2/p})$ neighbors in $S^\ast$. Note that each $v \in V_C \setminus S$ can learn all the relevant edges incident to the light nodes that can potentially form a $K_p$ with $v$.
Therefore, the only \emph{bad edges} that we cannot deal with are the edges contained in $S$. 
We will show that it is possible to assume that each cluster has at least $\Omega(n^{1 - 2/p})$ nodes, and so there are at most $O(n^{2/p})$ clusters. A calculation reveals that whenever $p \geq 5$, the set of bad edges over all clusters constitute at most a constant fraction of $E$.
By deferring dealing with these bad edges to subsequent iterations, we are done listing all instances of $K_p$ after $O(\log n)$ iterations. Section~\ref{section:Kp} contains the above proof, thus we obtain Theorem~\ref{theorem:Kp} for $p \geq 5$.

A different strategy is needed to deal with $K_4$, as in this case we cannot obtain a good bound on the number of bad edges.  We construct an expander decomposition recursively on the subgraph induced by the inter-cluster edges to cover all edges by high-conductance clusters. Now we only need to consider  $K_4=\{v_1, v_2, v_3, v_4\}$ instances crossing a cluster $C$ at the top-level expander decomposition and some other other cluster $C^\ast$ in the sense that $\{v_1, v_2\} \in E_C$ and $\{v_3, v_4\} \in E_{C^\ast}$.
We go over all possible pairs of $C$ and $C^\ast$ in parallel to let $C$ learn the edges in $C^\ast$ that can potentially form a $K_4$ with edges in $C$, by applying our above approach, replacing $V \setminus V_C$ by $V_{C^\ast} \setminus V_C$. 

The advantage of this new strategy is that it allows us to deal with the bad edges as follows.
It is possible to show that the bad edges in $C$ can be sent to $C^\ast$ efficiently. Moreover, the subgraph induced by the bad edges is sufficiently sparse that we can apply our sparsity-aware listing algorithm  to $C^\ast$ to list all the cross-cluster $K_4$ associated with the bad edges in  $\tilde{O}(n^{1 - 2/p})$ rounds.
Section~\ref{section:k4} contains the above proof, thus we obtain Theorem~\ref{theorem:Kp} for $p = 4$.

\subsection{Additional Related Work}

While our work solves $K_p$-listing for all $p \geq 4$ in optimal round complexity due to the lower bound of Fischer et al.~\cite{Fischer+SPAA18}, for the \emph{$K_p$-detection} problem,  the only lower bound known is due to Czumaj and Konrad~\cite{Czumaj+DISC18}, who showed that $\tilde{\Omega}(n^{1/2})$ rounds are needed for $K_p$ detection for all $4\leq p \leq n^{1/2}$ and that $\tilde{\Omega}(n/p)$ rounds are needed for $K_p$ detection for all $p \geq n^{1/2}$.

For cycles, Drucker et al.~\cite{Drucker+PODC14} showed that for fixed $p \geq 4$, $C_p$-detection requires $\Omega(\operatorname{ex}(n,C_p)/n))$ rounds, where $\operatorname{ex}(n,H)$ is the \emph{Tur\'{a}n number} that counts the maximum number of edges that an $n$-node graph can have without containing a subgraph isomorphic to $H$. Therefore, we have  a lower bound of $\tilde{\Omega}(n)$ for detecting $C_p$ when $p$ is odd, and we have a lower bound of $\tilde{\Omega}(n^{1/2})$ for the case $p=4$. 
 Korhonen and Rybicki~\cite{Korhonen+OPODIS17} extended this result to  make the $\tilde{\Omega}(n^{1/2})$ lower bound apply for all even $p$. They also showed that $C_p$-detection can be solved in $\tilde{O}(n)$  rounds for any constant $p$, implying that for constant odd values $p$ the complexity for $C_p$-detection is $\tilde{\Theta}(n)$. For even-length cycles, Fischer et al.~\cite{Fischer+SPAA18} showed that $C_{2p}$-detection can be solved in $O(n^{1-1/(p(p-1))})$ rounds. Eden et al.~\cite{Eden+DISC19} later improved this result to $\tilde{O}(n^{1-2/(p^2-p+2)})$ rounds for odd $p \geq 3$, and at most $\tilde{O}(n^{1-2/(p^2-2p+4)})$ rounds for even $p \geq 4$.
 Using expander decompositions, Eden et al.~\cite{Eden+DISC19} demonstrated a barrier to proving lower bounds for even-length cycle detection.
 There is a constant $\delta \in (0, 1/2)$ such that   any $\Omega(n^{(1/2)+\delta})$ lower bound on $C_{2p}$-detection would imply a new circuit lower bound.

Subgraph detection beyond cliques and cycles were also considered in the literature~\cite{Even+DISC17,Fischer+SPAA18,Gonen+OPODIS17,Korhonen+OPODIS17}.
Fischer et al.~\cite{Fischer+SPAA18} constructed a family  of graphs $H_p$ with $p$ nodes such that $H_p$-listing requires    $\Omega(n^{2-1/p} /p)$ rounds. Later, Eden et al.~\cite{Eden+DISC19} showed that for any $p$-node graph $H$,  the $H$-detection problem can be solved in $n^{2 - \Omega(1/p)}$ rounds, almost matching the lower bound of Fischer et al.~\cite{Fischer+SPAA18}.

Variants of subgraph finding problems requiring subgraphs to be reported by its constituent nodes were considered in~\cite{Abboud+arxiv17,Fischer+SPAA18,Huang+SODA20,Izumi+PODC17}. This additional requirement changes the nature of the problem drastically. For example,   the \emph{local} triangle listing problem, which requires  each $K_3$ to be reported by one of its three constituent nodes, requires $\tilde{\Omega}(n)$ rounds~\cite{Izumi+PODC17}, while the standard triangle listing problem has round complexity $\tilde{\Theta}(n^{1/3})$.
For the simpler problem that asks each node to decide whether it belongs to a triangle, 
any \emph{one-round} deterministic algorithm requires messages of 
 size $\Omega(\Delta \log n)$~\cite{Abboud+arxiv17}, matching the trivial upper bound. For the randomized setting, there is an $\Omega(\Delta)$ lower bound~\cite{Fischer+SPAA18} for the same problem.  From the upper bound side, Huang et al.~\cite{Huang+SODA20} showed that the \emph{local} triangle listing problem  can be solved in $O(\Delta / \log n + \log \log \Delta)$ rounds with high probability, matching the $\Omega(\Delta / \log n)$ lower bound of Izumi and Le~Gall~\cite{Izumi+PODC17} whenever $\Delta > \log n\log\log\log n$.

Quantum algorithms for distributed triangle detection have   been proposed recently by Izumi et~al.~\cite{Izumi+STACS20}, where they showed that triangle detection can be solved in $\tilde{O}(n^{1/4})$ rounds in the quantum version of \congest. This gives another example of a quantum algorithm beating the best known classical algorithms in distributed computing, as the current best known upper bound for distributed triangle detection in the classical setting is $\tilde{O}(n^{1/3})$~\cite{Chang+PODC19}.

\section{Preliminaries}\label{section:prelim}
We denote the set of neighbors of $v$ by $N(v)$.
For a node $v$ and a set of nodes $S$, we denote by $\deg_S(v)$ the number of neighbors that $v$ has in $S$.
Given a graph $G=(V,E)$ and two subsets $S,S'\subseteq V$, let $E(S,S')\subseteq E$ denote the set of edges with one extremity in $S$ and the other extremity in $S'$. 

The \emph{conductance} of a cut $(S, V \setminus S)$ is defined as $\Phi(S) = |\partial(S)| /  \min\{\vol(S), \vol(V \setminus S)\}$, where $\vol(U) = \sum_{v \in U} \deg(v)$, and  $\partial(S) = E(S, V\setminus S)$ is the set of edges between $S$ and $V \setminus S$. The conductance of a graph $G = (V, E)$, denoted by $\Phi(G)$, is defined as the minimum of the conductance of each cut in the graph, that is, $\Phi(G) = \min_{S \subseteq V}\Phi(S)$.
A \emph{lazy random walk} of a graph $G=(V,E)$ is a random walk on $V$ such that in each step, with probability $1/2$ it stays at the same node, and with probability $1/2$ it moves to a neighbor of the node, chosen uniformly at random.
We have the following relation~\cite{JerrumSICOMP89} between the 
mixing time $\mix(G)$ and conductance $\Phi(G)$: 
\[
\Theta\left(\frac{1}{\Phi(G)}\right) \leq \mix(G) \leq \Theta\left(\frac{\log n}{\Phi^2(G)}\right).
\]
Let $S$ be a node set. We write $G[S]$ to denote the subgraph induced by $S$, and we write $G\{S\}$ to denote the graph resulting from adding $\deg_V(v) - \deg_S(v)$ self loops to each node $v$ in  $G[S]$, where  each self loop of $v$ contributes 1 in the calculation of $\deg(v)$.
Note   that we always have
\[\Phi(G\{S\}) \leq \Phi(G[S]).\]

An \emph{$(\epsilon, \phi)$-expander decomposition} of a graph $G = (V,E)$ is  a partition of the node set $V = V_1 \cup V_2 \cup \cdots \cup V_k$ satisfying the following conditions.
\begin{itemize}
    \item For each cluster $V_i$, we have $\Phi(G\{V_i\}) \geq \phi$.
    \item The number of inter-cluster edges is at most $\epsilon |E|$.
\end{itemize} 

\begin{theorem}[Expander decomposition~\cite{Chang+PODC19}]\label{thm-expander-decomposition}
An $(\epsilon, \phi)$-expander decomposition with $\epsilon = 1/\polylog (n)$ and $\phi = 1/\polylog (n)$   can be constructed in $O(n^{0.001})$ rounds with high probability.
\end{theorem}

Note that we have the conductance guarantee not only for $G[C]$ but also for $G\{C\}$, for each cluster $C=(V_C,E_C)$ in the expander decomposition. In particular, each $v \in V_C$ with $|V_C| > 1$ must have at least an $1/\polylog (n)$ fraction of its neighbors in $C$, since otherwise $G\{C\}$ cannot have mixing time $\polylog (n)$.

\begin{theorem}[Expander routing~\cite{Chang+PODC19,Ghaffari+PODC17}]\label{thm-expander-routing}
Suppose $\mix(G) = \polylog (n)$.
There is an $O(n^{0.001})$-round algorithm that pre-processes the graph such that for any subsequent routing task where each node $v$ is a source and a destination of at most $L \cdot \deg(v)$ messages of $O(\log n)$ bits, all messages can be delivered in $L \cdot \polylog (n)$ rounds with high probability. 
\end{theorem}

To summarize, combining Theorem~\ref{thm-expander-decomposition} and Theorem~\ref{thm-expander-routing}, in $O(n^{0.001})$ rounds we can partition the node set $V$ into clusters with small mixing time $\polylog (n)$ and so the above routing task can be solved efficiently. Furthermore,
for a node $v \in V_C$ for any cluster $C=(V_C, E_C)$ with $|V_C| > 1$, it holds that $\deg_C(v) \geq \deg(v)/\polylog(n)$. 
 Note that for the trivial case of $|V_C| = 1$, we have $E_C = \emptyset$, and so there is no $K_p$ with edges in $C$. Therefore, in subsequent discussion, we only consider the  clusters with more than one node.

\section{Optimal Sparsity Aware Listing Algorithm}
\label{section:partition}
We present an algorithm which can be executed in a subgraph with good mixing time in order to optimally list all copies of $K_p$, for a given $p \geq 4$, with at least one edge in the subgraph, and potentially other edges outside of it. In order to show this algorithm, we begin by presenting a partitioning lemma, and then proceed to showing a theorem which performs sparsity aware listing.

\subsection{Input Partitioning}
We can modify the argument of~\cite[Lemma~4.2]{Chang+SODA19} in order to get the following statement. 

\begin{lemma}
\label{lemma:partition}
Let $G=(V,E)$ be a graph of maximum degree $\Delta$ where $V=V_1\cup V_2$ for two disjoints sets $V_1,V_2$ and $E= E_1\cup E_2\cup E_{12}$ for three sets $E_1\subseteq V_1\times V_1$, $E_2\subseteq V_2\times V_2$ and $E_{12}\subseteq V_1\times V_2$. Let $m_1$, $m_2$ and $m_{12}$ be upper bounds on the size of $E_1$, $E_2$ and $E_{12}$, respectively. Let $a$, $b$ and $\myn$ be three positive integers, with $a\le b$, satisfying the following conditions:
\begin{itemize}
\item[(a)]
$m_1\ge 20a|V_1|\log \myn$ and $m_1\ge 400 a^2\log^2 \myn$;
\item[(b)]
$m_2\ge 20b|V_2|\log \myn$ and $m_2\ge 400 b^2\log^2 \myn$;
\item[(c)]
$m_{12}\ge 20a|V_1|\log \myn$, $m_{12}\ge 20a|V_2|\log \myn$ and $m_{12}\ge 400 a^2\log^2 \myn$.
\end{itemize} 

 Assume that we create a partition $V_1^1,\ldots,V_1^a$ of $V_1$ as follows: each node $v\in V_1$ chooses uniformly at random a value $j\in\{1,\ldots, a\}$ and joins the set $V_1^j$. Similarly we create a partition $V_2^1,\ldots,V_2^b$ of $V_2$ as follows: each node $v\in V_2$ chooses uniformly at random a value $j\in\{1,\ldots, b\}$ and joins the set $V_2^j$.
Then with probability at least $1-\frac{30b^2\log\myn}{\myn^5}$ the following three statements hold:
\begin{itemize}
\item[(1)]
$|E(V_1^i,V_1^j)|\le 24 (m_1/a^2)$ for all $(i,j)\in[a]\times[a]$;
\item[(2)]
$|E(V_2^i,V_2^j)|\le 24 (m_2/b^2)$ for all $(i,j)\in[b]\times[b]$;
\item[(3)]
$|E(V_1^i,V_2^j)|\le 8 (m_{12}/a^2)$ for all $(i,j)\in[a]\times[b]$.\footnote{It is actually possible to obtain the stronger upper bound $|E(V_1^i,V_2^j)|=O(m_{12}/(ab))$ 
by imposing the stronger conditions $m_{12}\ge 20b|V_1|\log \myn$ and $m_{12}\ge 400 ab\log^2 \myn$. The present statement, which is slightly easier to prove, nevertheless suffices for our purpose.}
\end{itemize}
\end{lemma}
\begin{proof}
For any $(i,j)\in[a]\times[a]$, Statement~(1) holds with probability at least $1-(10\log \myn)/\myn^5$ from~\cite[Lemmas~4.2 and 4.3]{Chang+SODA19}. 
Similarly, for any $(i,j)\in[b]\times[b]$, Statement (2) holds with probability at least $1-(10\log \myn)/\myn^5$ as well. Let us show below that for any $(i,j)\in[a]\times[b]$, Statement~(3) holds with probability at least $1-(10\log \myn)/\myn^5$. Using the union bound then concludes the proof.

Let us fix $(i,j)\in[a]\times[b]$. For each edge $e\in E_{12}$ let $x_e$ denote the random variable that has value 1 if $e\in E(V_1^i,V_2^j)$ and value $0$ otherwise. Let us write $X=|E(V_1^i,V_2^j)|$. Observe that $X=\sum_{e\in E_{12}} x_e$ and $\E[X]=\frac{|E_{12}|}{ab}$. By Markov's inequality we have
\[
\Pr[X\ge 8 (m_{12}/a^2)]=\Pr[X^c\ge (8 (m_{12}/a^2))^c]\le \frac{\E[X^c]}{(8 (m_{12}/a^2))^c}=\frac{a^{2c}\E[X^c]}{8^cm_{12}^c}
\]
for any value $c> 0$. We choose the value $c=\floor{20\log\myn}$.

Now let us write 
\[
\E[X^c]=
\E\left[\Big(\sum_{e\in E_{12}} x_e\Big)^c\right]=
\sum_{k=0}^c\sum_{\ell=0}^c f_{k\ell}\frac{1}{a^{k}b^{\ell}}\le 
\sum_{k=0}^c\sum_{\ell=0}^c f_{k\ell}\frac{1}{a^{k+\ell}}
,
\]
where $f_{k\ell}$ denotes the number of $c$-tuples of edges from $E_{12}$ that have $k$ distinct points in $V_1$ and~$\ell$ distinct points in $V_2$.

We associate to each $c$-tuple of edges $(e_1,\ldots,e_c)\in E_{12}^c$ a vector $\vec{v}\in \{(0,0),(0,1),(1,0),(1,1)\}^c$ as follows: for each $s\in\{1,\ldots,c\}$, we set 
\[
\vec{v}_s=
\left\{
\begin{array}{ll}
(0,0) &\textrm{ if }e_s\in (V^i_1\setminus U_s)\times (V^j_2\setminus W_s),\\
(0,1) &\textrm{ if }e_s\in (V^i_1\setminus U_s)\times W_s,\\
(1,0) &\textrm{ if }e_s\in  U_s\times (V^j_2\setminus W_s),\\
(1,1) &\textrm{ if }e_s\in U_s\times W_s,
\end{array}
\right.
\]
where $U_s\subseteq V^i_1$ denotes the set of nodes from $V^i_1$ that are the endpoint of at least one edge in $\{e_1,\ldots,e_{s-1}\}$, and $W_s\subseteq V^j_2$ denotes the set of nodes from $V^j_2$ that are the endpoint of at least one edge in $\{e_1,\ldots,e_{s-1}\}$. 
We say that the vector $\vec{v}$ is of type $(d_{00},d_{01},d_{10},d_{11})$ if it contains $d_{00}$ times the coordinate $(0,0)$, $d_{01}$ times the coordinate $(0,1)$, $d_{10}$ times the coordinate $(1,0)$, and $d_{11}$ times the coordinate $(1,1)$, for positive integers $d_{00}$, $d_{01}$, $d_{10}$ and $d_{11}$ such that $d_{00}+d_{01}+d_{10}+d_{11}=c$. 

A given vector $\vec{v}$ is associated to more that one $c$-tuple in $E_{12}^c$.
A crucial observation is that there are at most 
\[
m_{12}^{d_{00}} (c|V_1|)^{d_{01}}(c|V_2|)^{d_{10}}c^{2d_{11}}
\]
$c$-tuples $(e_1,\ldots,e_c)\in E_{12}^c$ that are associated to a given vector $\vec{v}$ of type $(d_{00},d_{01},d_{10},d_{11})$. Indeed, when enumerating all the $c$-tuples that are associated to this $\vec{v}$, for each coordinate $(0,0)$ we can choose any edge that has not already been chosen, for each coordinate $(0,1)$ we have at most $|V_1|$ choices for the first endpoint and $c$ choices for the second endpoint, for each coordinate $(1,0)$ we have at most $c$ choices for the first endpoint and $|V_2|$ choices for the second endpoint, and for each coordinate $(1,1)$ we have at most $c$ choices for the first endpoint and $c$ choices for the second endpoint.

We now use this characterization to give an upper bound on $f_{k\ell}$. Let us write $x_{\min}=\max\{0,c-k-\ell\}$ and $x_{\max}=\min\{c-k,c-\ell\}$. Observe that the $c$-tuples of edges that contribute to $f_{k\ell}$ are those such the corresponding vector is of type $(k+\ell+x-c, c-\ell-x, c-k-x, x)$ for some value $x\in[x_{\min},x_{\max}]$. The argument of the previous paragraph enables us to give the following upper bound:
\begin{eqnarray*}
f_{k\ell}&\le&
\sum_{x=x_{\min}}^{x_{\max}}
c^{2x}
(c|V_1|)^{c-\ell-x}
(c|V_2|)^{c-k-x}
m_{12}^{k+\ell+x-c}\\
&\le& 
  a^{k+\ell-2c}m_{12}^c
\sum_{x=x_{\min}}^{x_{\max}}
\left(\frac{a^2c^{2}}{m_{12}}\right)^x,
\end{eqnarray*}
where we used the inequalities
$m_{12}\ge 20a|V_1|\log \myn$ and $m_{12}\ge 20a|V_2|\log \myn$ from Condition~(c), which implies $c|V_1|\le m_{12}/a$ and $c|V_2|\le m_{12}/a$ from our choice of $c$, to obtain the upper bound. Now, observing that $\frac{a^2c^{2}}{m_{12}}\le 1$ from Condition~(c), we get
\begin{eqnarray*}
f_{k\ell}&\le& 
 a^{k+\ell-2c}m_{12}^c(x_{\max}-x_{\min}+1)
\\
&\le& 
(c+1)
a^{k+\ell-2c}m_{12}^{c}.
\end{eqnarray*}

We thus obtain 
\[
\E[X^c]\le
\sum_{k=0}^c\sum_{\ell=0}^c 
(c+1)
\frac{m_{12}^c}{a^{2c}}
\le
(c+1)^3
\frac{m_{12}^c}{a^{2c}}.
\]

We conclude that 
\[
\Pr[X\ge 8 (m_{12}/a^2)]\le \frac{a^{2c}\E[X^c]}{8^cm_{12}^c}\le \frac{(c+1)^3}{8^c}
<\frac{10\log \myn}{\myn^{5}},
\]
as claimed.
\end{proof}

\subsection{Sparsity Aware Listing}
\label{sec1}
We show the following sparsity-aware $K_p$ listing algorithm for clusters in $G$ with low mixing time, see Figure~\ref{fig} for an illustration regarding $K_4$. Due to~\cite[Remark 2.6]{Censor+PODC20}\footnote{The remark states that since the lower bound for $K_{p = \omega(\log n)}$ listing is $\tilde{\Omega}(n)$, the case of $p = \omega(\log n)$ can be solved optimally (up to polylogarithmic factors) by a trivial algorithm.} throughout the paper we assume that $p = O(\log n)$, and so we can hide factors of $p$ under the $\tilde O(\cdot)$ notation.

\begin{theorem}\label{th:spars}
Let $C=(V_C,E_C)$ be any subgraph of $G = (V, E)$ with mixing time $\polylog(n)$. Let $\bar E\subseteq E(V_C,V \setminus V_C)$ be any set of edges so that any node $u\in V_C$ is incident to $\tilde O(\deg_{C}(u))$ edges from $\bar{E}$. 

Let $E'\subseteq E(V\setminus V_C, V\setminus V_C)$ be a set of edges given as input to~$C$ in the following way: the edges are distributed among the nodes of $C$ so that each node $u\in V_C$ receives $\tilde O(n^{1-2/p}\cdot \deg_{C}(u))$ edges from~$E'$.
If the conditions
\begin{equation}\label{ineq}
\frac{|E_C|}{|V_C|} = \Omega\left(\frac{|E'|}{n}\right),\  \frac{|E_C|}{|V_C|} = \Omega(n^{1/2}) 
\end{equation}
hold, then the graph $C$ can list in $\tilde O(n^{1-2/p})$ rounds, with high probability, by using only the edges $E_C$ for communication, all the $p$-cliques $\{v_1,\dots,v_p\}$ such that there exists some $p' \geq 2$, where
\begin{itemize}
\item
$V_{p'}=\{v_1,v_2, \dots, v_{p'}\} \subseteq V_C$, and $V_{p\setminus p'} = \{v_{p'+1}, \dots v_p\} \subseteq V \setminus V_C$,
\item 
$V_{p'} \times V_{p'} \subseteq E_C$, $V_{p'} \times V_{p\setminus p'} \subseteq \bar E$, and $V_{p\setminus p'} \times V_{p\setminus p'} \subseteq E'$.
\end{itemize}
\end{theorem}
\begin{figure}[t]
	\begin{center}
		\includegraphics[trim = 4cm 5.8cm 0cm 6.5cm, clip, scale=1]{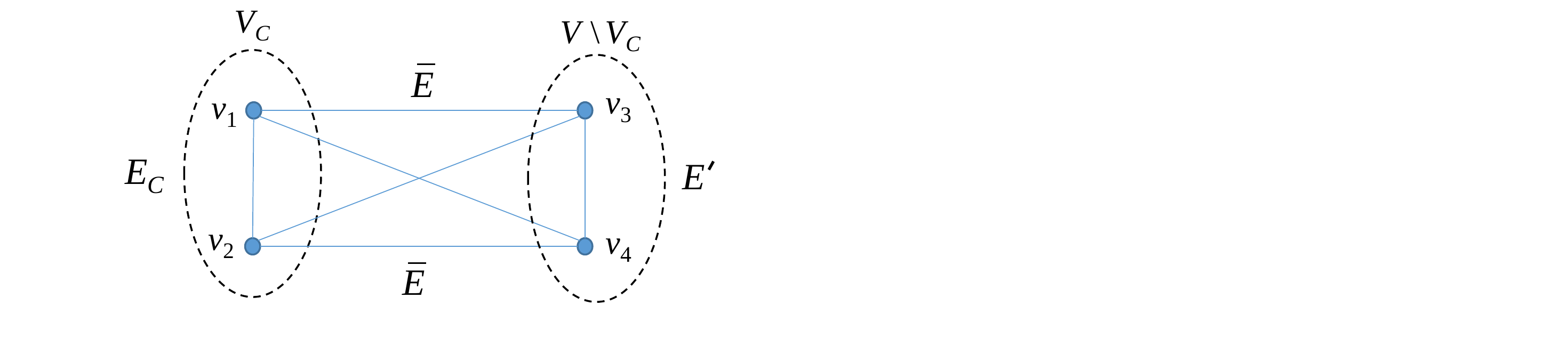}\vspace{-2mm}
	\end{center}
	\caption{\label{fig:partitionTree} Illustration for Theorem \ref{th:spars} with $p=4$ ($K_4$-listing) and $p'=2$.}\label{fig}
\end{figure}

\paragraph*{Remark.}
Theorem~\ref{th:spars} still works with the same proof if the lower bound requirements on the average degree $|E_{C}|/|V_{C}|$  of $C$ only hold when we restrict ourselves to the nodes in $V_C$ that have incident edges in $\bar E$. Intuitively, this is because that having a lot of \emph{irrelevant} low-degree nodes in $C$ does not reduce the capability for $C$ to list subgraphs, as they can be simply ignored for the case $2 \leq p' \leq p - 1$. For the special case of $p' = p$, we can simply apply the  subgraph listing algorithm  of~\cite[Theorem 5]{Chang+SODA19} with the routing algorithm of Theorem~\ref{thm-expander-routing} to $C$. In $\tilde{O}(n^{1 - 2/p})$ rounds all instances of $K_p$ in $C$ can be listed.

\begin{proof}
In what follows, we assume that  $|E'| \geq |E_C|$ and $k = |V_C| \geq n^{1-2/p}$. In the case that the former does not hold, we can always add $|E_C|$ \emph{dummy} edges to $E'$ without breaking any of the other conditions of the statement. In the case that the latter  does not hold, observe that each node $v \in V_C$ has at most $\tilde O(|V_C|) = \tilde O(k) = \tilde O(n^{1-2/p})$ neighbors $N(v)$ in the subgraph induced by $E_C \cup \bar E$.
Therefore, $v$ can list all required copies of $K_p$ involving $v$ in $\tilde O(n^{1-2/p})$ rounds by learning all edges in $N(v) \times N(v)$ as follows: $v$ sends the list $N(v)$ to all its neighbors $N(v)$, and then each $u \in N(v)$ sends $N(u) \cap N(v)$ to $v$.

We now describe the algorithm. We note that whenever we communicate between the nodes in $C$ in this algorithm, we utilize Theorem~\ref{thm-expander-routing}.

\paragraph*{Initialization.} We fix the value of $p' \in [p] \setminus \{1\}$, and perform the following algorithm which lists all copies of $K_p$ with exactly $p'$ nodes in $V_C$, $p-p'$ nodes in $V \setminus V_C$, and edges from the respective edge sets, $E_C, E'$, and $\bar E$. Thus, in order to list all copies of $K_p$ which satisfy the conditions of the theorem, we repeat this sequentially for all possible values of $p'$. By showing that for a specific $p'$ we can solve the problem in $\tilde O(n^{1-2/p})$ rounds, we achieve the required statement of the theorem. 

Denote by $k = |V_C|$, $m = |E_C|$, $n = |V|$, $m' = |E'|$. Denote the average degree in $C$ rounded down to the nearest power of 2 by $\delta = 2^{\floor{\log(2m/k)}}$. By the assumption that every node $u\in V_C$ is incident to at most $\tilde O(\deg_{C}(u))$ edges from $\bar{E}$, we get that $|\bar E| = \tilde{O}(|E_C|) = \tilde{O}(m)$. Clearly, since the mixing time of $C$ is $\polylog(n)$, the diameter of $C$ is $\polylog(n)$ and thus it is possible to trivially compute the values $k, m, m'$, and $\delta$, and to ensure that every node in $V_C$ knows them. Further, from the definition of the \congest model, we assume that all nodes know the value of $n$, and so all nodes in $V_C$ know all the values defined here.

We borrow the definition of \emph{vertex classes} from~\cite{Chang+SODA19}. Every node $v$ in $C$ computes $k_v = \deg_C(v)/\delta$. A node $v$ is in \emph{vertex class 0} if $k_v < 1/2$, and otherwise is in vertex class $i$ if $k_v \in [2^{i-2}, 2^{i-1})$. As shown in~\cite[Lemma 4.1]{Chang+SODA19}, it is possible in $\tilde O(1)$ rounds to reassign the IDs of all the nodes in $C$ such that the set of IDs is $[|C|]$ and the value $k_v$ can be computed from  $\ID(v)$.

Throughout the algorithm, we desire to \emph{only} use nodes in $C$ which have degree in $C$ at least half of the average, that is, nodes of class 1 and above, and so we denote these nodes by $C' \subseteq C$. Notice that $\sum_{v \in C'} \deg_C(v) \geq (1/2) \cdot \sum_{v \in C} \deg_C(v)$, and that $\sum_{v \in C'} 2k_v \geq k$. We now employ the nodes $C'$ to take responsibility for all the other nodes $C \setminus C'$. That is, each node $v \in C'$, of class $i$, is assigned some $2^{i+1}$ nodes $H_v \subseteq C \setminus C'$, such that each node in $C \setminus C'$ is assigned to exactly one node of $C'$. Notice that due to the definition of vertex classes, there exists a way to allocate the nodes of $C\setminus C'$ to the nodes of $C'$ while obeying these demands. Further, since every node knows the class of every other node, the nodes can locally compute these allocations. 

Now, node $v \in C'$ learns all the edges of $E_C, \bar E$ incident to any node in $H_v$, and also all the edges in $E'$ held by any node in $H_v$. Notice that, due to the constraints of this theorem, each node $u \in C$ holds at most $\tilde O(n^{1-2/p}\cdot \deg_C(u))$ such edges, and thus each node wishes to send and receive at most $\tilde O(n^{1-2/p}\cdot \deg_C(u))$ messages when every $v \in C'$ attempts to learn all the data held in the nodes $H_v$. As such, this step can be completed in $\tilde O(n^{1-2/p})$ rounds. Further, from now on, whenever we write that a node $v \in C \setminus C'$ attempts to send or receive a message, 
its corresponding node in $C'$ will be the one that actually sends or receives the message.
Notice that it is possible to do this while only increasing the round complexity of any algorithm by at most a poly-logarithmic factor, as each node $v \in C'$ now has at most $\tilde O(\deg_C(v))$ extra messages to send or receive per round. 

\paragraph*{Partitioning the graph.} We create two partitions $\mathbb{V}$ and $\mathbb{K}$, where $\mathbb{V}$ is a partition of $V \setminus V_C$ into $b = (k\cdot m'/m)^{1/p}$ roughly equally-sized parts, $\mathbb{V} = \{V_1, \dots, V_b\}$, and $\mathbb{K}$ is a partition of $V_C$ into $a = \sqrt{k}b^{1-p/2}$ roughly equally-sized parts, $\mathbb{K} = \{K_1, \dots, K_a\}$.\footnote{We assume that $a, b$ are integers, and otherwise round. Notice that as long as $1 < a, 1 < b$, this rounding can only incur an addition of a constant factor to the final round complexity. Nonetheless, we do need to show that $1 < a, 1 < b$. Notice that $1 < a$ if and only if $k > n^{1-2/p}$, which holds, as stated at the start of the proof, and that $1 < b$ holds due to the constraint that $|E'|>|E_C|$ and that $k > 1$.} To do so, every node $u \in V_C$ chooses uniformly at random which part in $\mathbb{K}$ to join, and sends the index of this part across all its incident edges in $E_C$, i.e., to all its neighbors in $C$. In order to create $\mathbb{V}$, we perform the following procedure. Let $v \in C'$ be some node chosen in an arbitrary, hardcoded way from $C'$. Notice that due to the conditions of this theorem, it must be that $\forall u \in C', \deg_C(u) = \Omega(n^{1/2})$. Node $v$ chooses uniformly at random a permutation $\pi$ on the nodes $V \setminus V_C$. Then, within $\tilde{O}(n^{1/2})$ rounds, node $v$ tells another, arbitrary node $u \in C'$ the value of $\pi$. We then repeat at most $\tilde O(1)$ iterations, where, in each iteration, each of the nodes in $C'$ which knows the value of $\pi$ tells this to some other node which does not know it, in a total of $\tilde{O}(n^{1/2})$ rounds per iteration. Now, all the nodes of $C'$ can locally compute $\mathbb{V}$ using $\pi$. Notice that since, as stated above, from here on the nodes of $C'$ simulate all the messages of nodes of $C \setminus C'$, we can implicitly assume that all of $C$ know $\mathbb{V}$, regardless of whether a node is in $C'$ or not.

We now utilize Lemma~\ref{lemma:partition} to claim that between any pair of parts in $\{V_1, \dots, V_b, K_1, \dots, K_a\}$, there are $\tilde O(n^{1-2/p} \cdot \delta)$ edges. We invoke Lemma~\ref{lemma:partition} with 
\begin{align*}
    m_1 &= \max \{m,\ 20ak \cdot \log n,\ 400a^2 \cdot \log^2 n\}, \\
    m_2 &= \max \{m', 20bn \cdot \log n,\ 400b^2 \cdot \log^2 n\}, \\
    m_{12} &= \max \{|\bar E|,\ 20ak \cdot \log n,\ 20an \cdot \log n,\ 400a^2 \cdot \log^2 n\}, \\
    \bar n &= n.
\end{align*}
We can see that all the conditions of the lemma hold, and thus we get that, with high probability, for every $i, j \in [a], i', j' \in [b]$, $|E(K_i, K_j)| \leq 24 \cdot m_1/a^2, |E(K_i, V_{j'})| \leq 8 \cdot m_{12}/a^2, |E(V_{i'}, V_{j'})| \leq 24 \cdot m_2 / b^2$. Thus, the number of edges between any pair of parts in $\{V_1, \dots, V_b, K_1, \dots, K_a\}$ is at most $O(m_1/a^2 + m_{12}/a^2 + m_2/b^2)$. Observe that the following hold
\begin{align*}
    O(m_1/a^2) &= \tilde O(m/a^2 + k/a + 1)\\
    O(m_2/b^2) &= \tilde O(m'/b^2 + n/b + 1)\\
    O(m_{12}/a^2) &= \tilde O(m/a^2 + k/a + n/a + 1) = \tilde O(m/a^2 + n/a + 1)\\
    O(m_1/a^2 + m_{12}/a^2 + m_2/b^2) &= \tilde O(m/a^2 + m'/b^2 + n/a + n/b + 1).
\end{align*} 

Primarily, notice that $a \leq b$ since $\sqrt{k}b^{1-p/2} \leq b \Longleftrightarrow \sqrt{k} \leq b^{p/2} = \sqrt{km'/m} \Longleftrightarrow m'\geq m$. Also, observe that $m/a^2 = m/(k \cdot b^{2-p}) = b^p \cdot m / (k \cdot b^2) = (m'\cdot k / m) \cdot m / (k \cdot b^2) = m'/b^2$. Thus, in combination with the above, the following is true
\begin{gather*}
    O(m_1/a^2 + m_{12}/a^2 + m_2/b^2) = \tilde O(m'/b^2 + n/a + 1).
\end{gather*}

Thus, we desire to show that $\tilde O(m'/b^2 + n/a + 1) = \tilde O(n^{1-2/p} \cdot \delta)$. Notice that due to the constraint that $\delta = \Theta(m/k) = \Omega(m'/n)$, we can see that the following holds
\begin{align*}
    \tilde O(m'/b^2) 
    &= \tilde O(m' (m' \cdot k /m)^{-2/p})\\
    &= \tilde O((m')^{1-2/p} \cdot (m/k)^{2/p}) \\
    &= \tilde O((m'/n)^{1-2/p} \cdot n^{1-2/p} \cdot (m/k)^{2/p}) \\
    &= \tilde O((m/k) \cdot n^{1-2/p}) \\
    &= \tilde O(\delta \cdot n^{1-2/p}).
\end{align*}

Due to the fact that $p \geq 4$ and $\delta = \Omega(n^{1/2})$, we obtain that 
\begin{gather*}
    \tilde O(n/a) = \tilde O(n) = \tilde O(\delta \cdot n^{1-2/p}).
\end{gather*}

Thus, we know that, with high probability, the number of edges between any two parts in $\{V_1, \dots, V_b, K_1, \dots, K_a\}$ is $\tilde O(n^{1-2/p}\cdot \delta)$.

\paragraph*{Reshuffling the input.}
Before we perform the final listing step, we need to reshuffle the edges of $E'$ across the nodes of $C$. In the following part (\emph{Performing the Listing}), every edge of $E'$ is broadcast to many nodes in the graph, and, due to symmetry, every edge is broadcast to the same number of nodes. We thus need to reshuffle these edges across $C$ such that nodes with higher degree hold more edges. Specifically, we desire for each node $v \in C'$ to hold $O(m' \cdot k_v/k)$ edges, and since, as seen before, $\sum_{v \in C'}2k_v \geq k$, this covers all the edges in $E'$. Using the fact that the class of a node $v \in C'$, the value $\floor{\log k_v}-2$, can be computed locally by all the nodes, implying that a $2$-approximation to the degree of every node is known to every other node, the nodes can locally know how many edges each node needs to \emph{receive}. In a similar manner, the nodes construct \emph{additional} vertex classes in order to know $2$-approximations to the number of edges in $E'$ which are held by a node. This information allows the nodes to compute how many edges each needs to \emph{send}. Combined, these parts suffice for every node to know which node it needs to send edges to. 
Since each node $v \in C'$ originally holds at most $\tilde O(n^{1-2/p}\cdot \deg_C(v))$ edges of $E'$, it is possible to perform the reshuffling within $\tilde O(n^{1-2/p})$ rounds.

\paragraph*{Performing the listing.}
We are now arriving at the final stage of the algorithm, where nodes in $C$ are assigned various parts in $\{V_1, \dots, V_b, K_1, \dots, K_a\}$ and are required to learn all the copies of $K_p$ between those parts. Primarily, notice that since we know that $a \leq b$, the total number of ways to choose $p'$ parts from $\{K_1, \dots, K_a\}$ and $p-p'$ parts from $\{V_1, \dots, V_b\}$ is bounded by $a^{p'}\cdot b^{p-p'} \leq a^2\cdot b^{p-2}=k$. Further, as done in~\cite{Chang+SODA19}, since $\sum_{v \in C'}2k_v \geq k$, it is possible to assign, in a hardcoded, globally known manner the $k$ choices of $p$ parts to the nodes $C'$ such that each $v\in C'$ receives between $2k_v$ and $4k_v$ sets of $p$ parts, as all the nodes can compute a $2$-approximation of $k_v$ for any $v \in C'$. Finally, each node $v \in C'$ desires to learn the edges between all the parts which it is assigned, and thus receive at most $\tilde O(n^{1-2/p} \cdot \delta \cdot k_v) = \tilde O(n^{1-2/p} \cdot \deg_C(v))$ messages. Likewise, we desire that every node send at most such many messages. To achieve that, notice primarily that all the nodes in $C$ know, given an edge in the graph, which nodes in $C'$ need to receive it as part of the listing. Further, notice that within every edge set $E_C, E', \bar E$, every edge is required to be sent to exactly the same number of nodes in $C'$. Thus, since the total amount of information which needs to be sent over the entire graph $C$ is $\tilde O(m\cdot n^{1-2/p})$ edges (since this bounds the number of messages received), each edge in $E_C$, $\bar E$, and $E'$ is sent to at most $\tilde O(n^{1-2/p})$, $\tilde O(n^{1-2/p})$, and $\tilde O(n^{1-2/p}\cdot m/ m')$ nodes, respectively. Due to the fact that each node $v \in C$ originally has $\tilde O(\deg_C(v)))$ edges incident to it in $E_C$ and $\bar E$, this implies that sending the first two types of edges incurs $\tilde O(n^{1-2/p}\cdot \deg_C(v))$ messages from $v$. Further, due to the reshuffling step, all of $E'$ is stored in the nodes $C'$, where $v \in C'$ stores $O(m' \cdot k_v/k)$ edges, and thus sending each to $\tilde O(n^{1-2/p}\cdot m/ m')$ nodes, incurs a total of $\tilde O((m' \cdot k_v/k) \cdot n^{1-2/p}\cdot m/ m') = \tilde O(n^{1-2/p} \cdot k_v \cdot (m/k)) = \tilde O(n^{1-2/p} \cdot \deg_C(v))$ messages sent from~$v$.

As the total number of messages sent and received by node $v \in C'$ is at most $\tilde O(n^{1-2/p}\cdot \deg_C(v))$, we conclude that the round complexity of this final stage of the algorithm is $\tilde O(n^{1-2/p})$ as well. 
\end{proof}

\section{Optimal $K_p$-listing Algorithm for $p>4$}
\label{section:Kp}

We show here how to list all instances of $K_p$, for $p>4$.

We start with a simple procedure, in which each node of small degree lists all cliques that it is a part of by an \emph{exhaustive search} approach. Formally, each node $v$ with $\deg(v)\leq 2n^{1/2}$ sends its neighborhood $N(v)$ to all of its neighbors in $N(v)$. Each neighbor sends an ack about each neighbor, and so $v$ learns about all edges in $N(v) \times N(v)$, thereby $v$ lists all $K_p$ instances involving it. Thus we remove from the graph the node $v$ and all edges touching $v$. This clearly takes at most $O(n^{1/2})$ rounds, according to the degree condition for $v$. 

This procedure does not give us any promise on the degrees given by the remaining edges because by removing edges we could now have additional low-degree nodes. However, the property that this guarantees is that this handles many nodes in case the average degree in the graph was initially not too large. Denote by $\avgdeg$ the average degree of the graph.  

\begin{claim}
\label{claim:lowDeg}
Listing all instances of $K_p$ involving $v$ such that $\deg(v) \leq 2n^{1/2}$ can be done within $O(n^{1/2})$ rounds by an exhaustive search procedure. This can be done in parallel for all such nodes~$v$.
Further, if $\avgdeg < n^{1/2}$, then this removes at least half of the nodes from the graph.
\end{claim}

\begin{proof}
The fact that we can list all such instances within $O(n^{1/2})$ rounds is straightforward by exhaustive search.

Assume the case in which $\avgdeg < n^{1/2}$. If the number of nodes with degree at most $2n^{1/2}$ is less than $n/2$, then the average degree $\avgdeg$ is more than $(2n^{1/2} \cdot n/2 )/n$, which contradicts the assumption.
\end{proof}

\sloppy{After the above procedure, the algorithm proceeds as follows. We run the $(1/\polylog(n),1/\polylog(n))$-expander decomposition of~\cite{Chang+PODC19}, as stated in Theorem~\ref{thm-expander-decomposition}. Denote by $E_m$ the edges within clusters and by $E_r$ the remaining edges. Recall that we have the following properties: 
}
\begin{enumerate}
    \item\label{item:intercluster} $|E_r| \leq |E|/\polylog(n)$.
    \item\label{item:mixing} The mixing time within each cluster $C$ is $O(\polylog(n))$.
    \item\label{item:deg} For a node $v \in V_C$ for any cluster $C=(V_C, E_C)$, it holds that $\deg_C(v) \geq \deg(v)/\polylog(n)$.
\end{enumerate}

We now consider two cases, depending on the size $|V_C|$ of a cluster.

Let $\beta>1$ be some constant threshold which we will fix later.
Suppose $|V_C| < \beta n^{1-2/p}$, and consider a node $v \in V_C$. In this case we again follow an exhaustive search approach for $v$: The node~$v$ learns all of $N(v) \times N(v)$ by sending $N(v)$ to all nodes in $N(v)$ and receiving an ack from each recipient about each neighbor. This implies that all instances of $K_p$ involving $v$ are listed by~$v$.

\begin{claim}
\label{claim:smallC}
Let $\beta>1$ be some constant and let $C$ be a cluster such that $|V_C| < \beta n^{1-2/p}$ and consider a node $v\in V_C$.
Then, listing all instances of $K_p$ involving $v$ can be done within $\tilde{O}(n^{1-2/p})$ rounds. This can be done in parallel for all such clusters $C$ and nodes $v$.
\end{claim}

\begin{proof}
By Property~\ref{item:deg} above, we have that $\deg(v) \leq \deg_{C}(v)\cdot\polylog(n)$. Hence, $|N(v)| \leq |V_C|\cdot\polylog(n)$ and the bound on $|V_C|$ implies that $|N(v)| \leq \tilde{O}(n^{1-2/p})$. Thus, sending $N(v)$ to all nodes in $N(v)$ completes in $\tilde{O}(n^{1-2/p})$ rounds. Receiving the acks takes another $\tilde{O}(n^{1-2/p})$ (although up to a single round this can be done concurrently with the sending). Finally, notice that for each such~$v$ the communication only takes place with $v$  itself, and so this can be done in parallel for all such $C$ and $v$.
\end{proof}

In particular, Claim~\ref{claim:smallC} leaves us only with sufficiently large clusters. Hence, from now on we suppose $|V_C| \geq \beta n^{1-2/p}$.

Consider the nodes in $V\setminus V_C$, and denote $$S^{*}_C =\{ u \not\in V_C \mid 1 \leq \deg_{C}(u) < \deg_{V\setminus V_C}(u) / n^{1-2/p} \}.$$ Note that every $u \in S^{*}_C$ has a lower bound of $n^{1-2/p}$ on its degree. Also, for every $u \in S^{*}_C$ it holds that $\deg_C(u) \leq n^{2/p}$, and thus the total number of edges between nodes in $C$ and nodes in $S^{*}_C$ is at most  $n^{1+2/p}$.  We say that a node $v \in V_C$ is a \emph{bad node} if it has more than $n^{1-2/p}$ neighbors in $S^{*}_C$. These nodes are denoted by $S_C = \{v\in V_C \mid \deg_{S^{*}_C}(v)> n^{1-2/p}\}$. 

\begin{claim}
\label{claim:bad2largeC}
It holds that $\sum_{C \text{ such that } |V_C| \geq \beta n^{1-2/p}}{|S_C|^2} \leq (4/\beta)|E|$, where $\beta>1$ is a constant.
\end{claim}

\begin{proof}
Clearly, the total number of edges $|E|$ is at least the number of edges that touch nodes in $S^{*}_C$. We bound this from below, and we do so by only counting edges that touch nodes in $S^{*}_C$ but \emph{do not} touch nodes in $V_C$, by summing $\deg_{V\setminus V_C}(u)$ over all $u \in S^{*}_C$ and dividing by two, due to possible double counting. This gives that 
\begin{align*}
2|E| &\geq \sum_{u \in S^{*}_C}{\deg_{V\setminus C}(u)} > n^{1-2p}\sum_{u \in S^{*}_C}{\deg_{C}(u)} \\
&\geq n^{1-2p}\sum_{v \in S_C}{\deg_{S^{*}_C}(v)} > n^{1-2p}\cdot n^{1-2p} \cdot |S_C| =  n^{2-4p} |S_C|.
\end{align*}
This implies that $|S_C| \leq 2|E|/n^{2-4/p}$. Also, since $|E| \leq n^2$, this implies that $|S_C| \leq 2n^{4/p}$. Thus, the total number of edges within $S_C$ is at most $|S_C| \cdot |S_C| \leq 2n^{4/p} \cdot 2|E|/n^{2-4/p} = 4|E|/n^{2-8/p}$. Since there are at most $n^{2/p}/\beta$ clusters of size $|V_C| \geq \beta n^{1-2/p}$, we have that the total number of edges within $S_C$ over all such clusters is at most $(n^{2/p}/\beta) \cdot (4|E|/n^{2-8/p}) \leq (4/\beta) |E|/ n^{2-10/p}$. For every $p \geq 5$, this is at most $(4/\beta)|E|$.
\end{proof}

We now show that within $\tilde{O}(n^{1-2/p})$ rounds we can list all instances of $K_p$ which have an edge in $(V_C\times V_C) \setminus (S_C \times S_C)$ for some of the clusters $C$. Along with the edges in $E_r$ and the edges in $S_C \times S_C$, there will still be some additional clusters for which we will need to defer their edges to following iterations.

First, notice that for any $v\in V_C\setminus S_C$ we have $\deg_{S^{*}_C}(v) \leq n^{1-2/p}$, and thus $v$ can learn about all edges in $(N(v) \cap S^{*}_C) \times (N(v) \cap S^{*}_C)$  by an exhaustive search of sending $N(v) \cap S^{*}_C$ to all nodes in $N(v) \cap S^{*}_C$ and receiving an ack from each recipient about each neighbor. Similarly to how it is done
in Claim~\ref{claim:smallC}, we can do this in parallel for all such $C$ and $v$, and it completes within $O(n^{1-2/p})$ rounds. 

Second, notice that for every $u \not\in V_C$ for which $u\not\in S^{*}_C$, it holds that  $\deg_{C}(u) \geq \deg_{V\setminus C}(u) / n^{1-2/p} $. Thus, $u$ can make each of its edges be known to some node in $V_C$ by batching its edges into non-overlapping chunks of at most $n^{1-2/p}$ edges, and sending each chunk to a different neighbor of $u$ in $C$. This completes in within $n^{1-2/p}$ rounds, in parallel for all such $C$ and $u$.

We now have that for each $v \in V_C\setminus S_C$, all edges between its neighbors are known to some nodes in $V_C$. It remains to show how each cluster lists the instances of $K_p$ that are contained in the set of edges known to it. 

Our goal is now to utilize our sparsity-aware listing algorithm given in Theorem~\ref{th:spars}, for which we need some good lower bounds on the average degree of the cluster. To this end, we simply defer all clusters whose average degree is too small to the next iteration. Formally, as in Theorem~\ref{th:spars}, denote by $E'\subseteq E(V\setminus V_C,V\setminus V_C)$ the edges that do not touch the nodes of $C$ that are now distributed among the nodes of $C$. The cluster $C$ computes $|E'|$ and $|E_C/V_C|$ within $\polylog(n)$ rounds, by Theorem~\ref{thm-expander-routing}.

Let $\gamma,\gamma'>1$ be two additional constants which we will fix later.
If $|E_C|/|V_C| \leq  |E'|/\gamma n$ or if $|E_C|/|V_C| <  n^{1/2}/\gamma'$, then we say that $C$ is a \emph{low-average cluster}. Let $E_{\operatorname{low}}$ be the set of all edges in low-average clusters (i.e., $E_{\operatorname{low}}$ is the union of $E_C$ for all low-average clusters $C$).
All edges in $E_{\operatorname{low}}$
are deferred to the next iteration, and we claim that this defers only another small constant fraction of the edges, if the average degree in the graph is sufficiently large. 
\begin{claim}
\label{claim:lowAvgC}
If $\avgdeg\geq n^{1/2}$, then $|E_{\operatorname{low}}| \leq (1/\gamma+1/\gamma') |E|$.
\end{claim}

\begin{proof}
Recall that $\avgdeg$ denotes the average degree in the graph. Since $|E'|\leq |E|$, by considering all clusters $C$ that satisfy $|E_C|/|V_C| < |E'|/\gamma n$, we consider nodes with average degree inside their cluster bounded from above by $\avgdeg/\gamma$, therefore the number of edges within all such clusters is at most $\avgdeg n/2\gamma$ edges. However, 
by Property~\ref{item:intercluster} of the decomposition, we know that the number of edges within the clusters (inside $E_m$) is at least $|E|/2$, which is at least $\avgdeg n/4$. Therefore, the total number of edges in all clusters for which $|E_C|/|V_C| < |E'|/\gamma n$ is at most $\avgdeg n/2\gamma = (2/\gamma)\avgdeg n/4 \leq (2/\gamma)|E|/2 = (1/\gamma)|E|$.

Similarly, if $\avgdeg\geq n^{1/2}$, by considering all low-average clusters $C$ for which $|E_C|/|V_C| <  n^{1/2}/\gamma'$, we consider nodes with average degree inside the cluster bounded from above by $\avgdeg/\gamma'$, therefore the total number of edges in all such clusters is at most $\avgdeg n/2\gamma'$. 

Therefore, the total number of edges in all clusters for which $|E_C|/|V_C| <  n^{1/2}/\gamma'$ is at most $\avgdeg n/2\gamma' = (2/\gamma')\avgdeg n/4 \leq (2/\gamma')|E|/2 = (1/\gamma')|E|$.

To summarize, we obtain that $|E_{\operatorname{low}}| \leq (1/\gamma)|E| + (1/\gamma')|E| = (1/\gamma+1/\gamma')|E|$, as claimed.
\end{proof}

We show that the conditions of the theorem hold for any cluster $C=(V_C,E_C)$ with $|V_C|\geq \beta n^{1-2/p}$ which is \emph{not} a low-average cluster. 

First, by Property~\ref{item:mixing} above, we have that the mixing time of $C$ is $O(\polylog(n))$. Furthermore, by Property~\ref{item:deg}, we have that for every $u\in V_C$ it holds that $\deg(u)$ is at most $\tilde{O}(\deg_C(u))$, which implies that $\deg_{\bar{E}}(u)$ is also bounded by $\tilde{O}(\deg_C(u))$, where $\bar{E}$ is a subset of $E(V_C,V\setminus V_C)$. In addition, if we consider the set of edges $E'\subseteq E(V\setminus V_C,V\setminus V_C)$ that are distributed among the nodes of $C$, then indeed each node $u\in V_C$ holds no more than $O(n^{1-2/p}\cdot \deg_C(u))$ of those edges. The reason is that each edge of $E'$ that reaches $u$, reaches it through one of its edges in $E(V_C,V\setminus V_C)$. There are at most $\tilde{O}(\deg_C(u))$ edges in the latter, and only $O(n^{1-2/p})$ rounds, which implies the required bound on the amount of information that they convey to $u$.

Finally, we note that the bounds required by Equation~\ref{ineq} hold for $C$ which is not a low-average cluster, by definition. 

Since all required conditions hold, we now execute the algorithm provided by Theorem~\ref{th:spars} on each cluster $C$ with $|V_C|\geq \beta n^{1-2/p}$ which is not a low-average cluster. We then remove all edges $(V_C\times V_C) \setminus (S_C \times S_C)$ in each such cluster $C$ and continue to the next iteration with all edges in $E_r \cup \bigcup_{C \text{ such that } |V_C| \geq \beta n^{1-2/p}}(S_C \times S_C) \cup E_{\operatorname{low}}$.

By Property~\ref{item:intercluster} above, Claim~\ref{claim:bad2largeC}, and Claim~\ref{claim:lowAvgC}, there are no more than $(1/\polylog(n)+4/\beta+ 1/\gamma +1/\gamma')|E|$ remaining edges if the average degree $\avgdeg$ of the graph is at least $\avgdeg \geq n^{1/2}$. Choosing $\beta=32$, $\gamma=8$, and $\gamma'=8$, gives that there are at most $4(1/8)|E|=|E|/2$ remaining edges.

In case the average degree $\avgdeg$ of the graph is bounded by $\avgdeg < n^{1/2}$, the initial exhaustive search procedure for nodes of degree at most $2n^{1/2}$ removes at least half of the nodes by Claim~\ref{claim:lowDeg}.

This implies that we complete within a logarithmic number of iteration, which proves that within  $\tilde{O}(n^{1-2/p})$ rounds we list all instances of $K_p$, for $p>4$.

\section{Optimal $K_4$-listing Algorithm}
\label{section:k4}

We follow the same high-level framework of the algorithm for listing $K_p$ for $p > 4$. That is, we deal with all nodes $v$ with $\deg(v) \leq 2n^{1/2}$ by an exhaustive search approach in $O(\sqrt{n})$ rounds and remove them from the graph. Then we compute an expander decomposition using  Theorem~\ref{thm-expander-decomposition}. Recall that if the average degree $\mu$ at the beginning is at most $n^{1/2}$, then at least half of the nodes will be removed by   Claim~\ref{claim:lowDeg}.

After that, we apply Claim~\ref{claim:smallC} to deal with all clusters $C$ with $|V_C| < \beta n^{1 - 2/p} = \beta n^{1/2}$ in $\tilde{O}(n^{1/2})$ rounds, and then the nodes in these clusters and their incident edges are removed.  Denote~$\Cset$ the set of remaining clusters. Note that each $C \in \Cset$ has $|V_C| \geq \beta n^{1/2}$, and so the number of these clusters is $|\Cset| = O(n^{1/2})$.

What makes the case $p=4$ different from the case $p > 4$ is that Claim~\ref{claim:bad2largeC} does not hold when $p = 4$. To deal with this issue, we will consider a different approach to listing cross-cluster $K_4$. We cover all edges in the graph by clusters with small mixing time by recursively computing an expander decomposition of the subgraph induced by the inter-cluster edges for $O(\log n)$ iterations. Denote $\Cstarset$ as the union of $\Cset$ and the set of clusters in these expander decompositions.
The following claim allows us to process clusters in $\Cstarset$ in parallel with small overhead.

\begin{claim}\label{claim:smalloverlap}
Each node $v$ belongs to at most $O(\log n)$ distinct clusters in $\Cstarset$.
\end{claim}
\begin{proof}
This follows from the fact that the clusters in $\Cstarset$ are from $O(\log n)$ distinct expander decompositions, and in each expander decomposition the clusters are node-disjoint.
\end{proof}

For all $K_4$ that completely reside in some cluster $C \in \Cset$, we can list all of them by applying Theorem~\ref{th:spars} with $\bar E = \emptyset$, $E' = \emptyset$ and $p=4$ to $C$. This costs $\tilde{O}(n^{1 - 2/p}) = \tilde{O}(\sqrt{n})$ rounds.

Now, for each remaining  $K_4=\{v_1,v_2,v_3,v_4\}$ with at least one edge $\{v_1,v_2\} \in E_C$ in a cluster $C \in \Cset$, we must have $\{v_3,v_4\} \in E_{C^\ast}$ for some  cluster $C^\ast \in \Cstarset \setminus\{C\}$. Our strategy for  listing those cross-cluster $K_4$ is to go over all pairs of clusters $C \in \Cset$ and $C^\ast \in \Cstarset \setminus\{C\}$ and to transmit edges between $C$ and $C^\ast$ in such a way that each $K_4$ crossing $C$ and $C^\ast$ is learned by one of $C$ and $C^\ast$ so that Theorem~\ref{th:spars} can be applied to list them.  For any distinct clusters $C$ and $C^\ast$, define the sets
\begin{align*}
\Sstar_{C^\ast\to C}&= \{u\in V_{C^\ast} \:|\:1\le \deg_C(u) <\deg_{C^\ast}(u)/ \sqrt{n}\},\\
S_{C\to C^\ast}&=\{u\in V_{C}\:|\:\deg_{\Sstar_{C^\ast\to C}}(u)> \sqrt{n}\},
\end{align*}
which are analogous to the two sets  $\Sstar_C$  and $S_C$ defined in Section~\ref{section:Kp}.

Our algorithm consists of three parts.

\paragraph*{First part.} The first part of the algorithm is as follows.
\begin{itemize}
\item[1.]
Each cluster $C^\ast  \in \Cstarset$ does as follows. For each $C \in \Cset \setminus \{C^\ast\}$, each node $u\in V_{C^\ast}\setminus \Sstar_{C^\ast\to C}$ such that $\deg_C(u)\ge 1$ 
sends the set $N_{C^\ast}(u)$ to $C$, which can be done in $\tilde{O}(\sqrt{n})$ rounds, due to the inequality $\deg_{C}(u)\ge \deg_{C^\ast}(u)/\sqrt{n}$ and Claim~\ref{claim:smalloverlap}. 
\item[2.]
Each cluster $C \in \Cset$ then does as follows. Similar to the case of $p>4$, denote by $E'\subseteq E(V\setminus V_C,V\setminus V_C)$ the edges that do not touch the nodes of $C$ that are now distributed among the nodes of $C$ due to Step~1. The cluster $C$ computes $|E'|$ and $|E_C/V_C|$ within $\polylog(n)$ rounds, by Theorem~\ref{thm-expander-routing}.
Recall that $C$ is a low-average cluster if $|E_C|/|V_C| \leq  |E'|/\gamma n$ or if $|E_C|/|V_C| <  n^{1/2}/\gamma'$. All the edges $E_C$ in each low-average cluster $C$ are deferred to the next iteration. By Claim~\ref{claim:lowAvgC},  this defers only a small constant fraction of the edges, if the average degree at the beginning satisfies $\mu \geq n^{1/2}$.
Now suppose $C$ is not low-average w.r.t.~$E'$.  Let $\bar E = E(V_C,V \setminus V_C)$, and then the cluster $C$ uses Theorem~\ref{th:spars} to list all $K_4$ in the subgraph induced by  $E' \cup E_C \cup \bar E$ that have at least one edge in $E_C$, and this takes $\tilde{O}(\sqrt{n})$ rounds.  Note that since $C$ is not low-average, all the required conditions in Theorem~\ref{th:spars} are met.
\end{itemize}

This first part enables each   cluster $C \in \Cset$ that is not low-average w.r.t.~$E'$ to list all 4-cliques $\{v_1,v_2,v_3,v_4\}$ such that $v_3\in V_{C^\ast}\setminus \Sstar_{C^\ast\to C}$ or $v_4\in V_{C^\ast}\setminus \Sstar_{C^\ast\to C}$ holds for some $C^\ast \in \Cstarset \setminus \{C\}$.

\paragraph*{Second part.}
The second part of the algorithm is as follows.
\begin{itemize}
\item[1.] 
Each cluster $C \in \Cset$ does as follows. For each $C^\ast \in \Cstarset \setminus \{C\}$, each node $u\in V_C\setminus S_{C\to C\ast}$ learns all edges in $N_{\Sstar_{C^\ast\to C}}(u)\times N_{\Sstar_{C^\ast\to C}}(u)$ from the cluster $C^\ast$. This can be done in $\tilde{O}(\sqrt{n})$  rounds due to the definition of $S_{C\to C^\ast}$ and Claim~\ref{claim:smalloverlap}.

\item[2.]
Each cluster $C \in \Cset$ then does as follows. Similar to the first part, let $\bar E = E(V_C,V \setminus V_C)$, and denote by $E'\subseteq E(V\setminus V_C,V\setminus V_C)$ the edges that do not touch the nodes of $C$ that are now distributed among the nodes of $C$ due to Step~1. 
In case $C$ is not low-average w.r.t.~$E'$, it uses Theorem \ref{th:spars} to list all $K_4$ in the subgraph induced by  $E' \cup E_C \cup \bar E$ that has at least one edge in $E_C$, and this takes $\tilde{O}(\sqrt{n})$ rounds.
\end{itemize} 

This second part enables each cluster $C \in \Cset$ that is not low-average w.r.t.~$E'$ to list all the cliques $\{v_1,v_2,v_3,v_4\}$ such that the following conditions hold for some $C^\ast \in \Cstarset \setminus \{C\}$:
\begin{itemize}
    \item $v_1\in V_C\setminus S_{C\to C^\ast}$ or $v_2\in V_C\setminus S_{C\to C^\ast}$,
    \item $v_3, v_4 \in \Sstar_{C^\ast\to C}$.
\end{itemize}

\paragraph*{Third part.}
The remaining task is listing cliques $\{v_1,v_2,v_3,v_4\}$ such that $v_1,v_2\in S_{C\to C^\ast}$ and $v_3,v_4\in \Sstar_{C^\ast\to C}$ hold for some clusters $C \in \Cset$ and $C^\ast \in \Cstarset \setminus \{C\}$. 
The third part of the algorithm performs this task as follows.  

\begin{itemize}
\item[1.]
Each  cluster $C \in \Cset$ 
does as follows. For each $C^\ast \in \Cstarset \setminus \{C\}$, each node $u\in S_{C\to C^\ast}$ sends the set of edges $\{ \{u,v\}\:|\:v\in N_{S_{C\to C^\ast}}(u) \}$ to $C^\ast$ by dividing this set into $\deg_{\Sstar_{C^\ast\to C}}(u)$ subsets of roughly the same size and sending each subset to one of its neighbors in $\Sstar_{C^\ast\to C}$. This can be implemented in $\tilde{O}(\sqrt{n})$ rounds, by the definition of $S_{C\to C^\ast}$ and Claim~\ref{claim:smalloverlap}.
\item[2.]
Each cluster $C^\ast \in \Cstarset$ does as follows.
Denote by $E'\subseteq E(V\setminus V_C,V\setminus V_C)$ the edges that do not touch the nodes of $C^\ast$ that are now distributed among the nodes of $C^\ast$ due to Step~1, and let $\bar E =\bigcup_{C \in \Cset \setminus \{C^\ast\}} E(S_{C\to C^\ast}, \Sstar_{C^\ast\to C})$. The cluster $C^\ast$ uses Theorem~\ref{th:spars} to list all $K_4$ in the subgraph induced by  $E' \cup E_{C^\ast} \cup \bar E$ that has at least one edge in $E_{C^\ast}$, and this takes $\tilde{O}(\sqrt{n})$ rounds.

\end{itemize}

The following claims show that all conditions of Theorem~\ref{th:spars} are satisfied. 
\begin{claim}
For each $C^\ast \in \Cstarset$, each node $v\in C^\ast$ is incident to $O(\deg_{C^\ast}(v))$ edges from $\bar{E}$. 
\end{claim}
\begin{proof}
For each node $v\in V_{C^\ast}$ and any cluster $C \in \Cset \setminus \{C^\ast\}$ such that $v\in \Sstar_{C^\ast\to C}$, we have $\deg_C(u) <\deg_{C^\ast}(u)/ \sqrt{n}$. 
The number of edges from $\bar{E}$ incident to any $v\in V_{C^\ast}$ is thus at most
\begin{align*}
\sum_{C \in \Cset \setminus \{C^\ast\}, \:v\in \Sstar_{C^\ast\to C}}\deg_{S_{C\to C^\ast}}(v)&\le O(\sqrt{n}) \cdot \max_{C \in \Cset \setminus \{C^\ast\}, \:v\in \Sstar_{C^\ast\to C}} \{\deg_{S_{C\to C^\ast}}(v)\}\\
&\le O(\deg_{C^\ast} (v)),
\end{align*}
where the first inequality follows from the fact that $|\Cset| = O(\sqrt{n})$, and the second inequality follows from the definition of  $\Sstar_{C^\ast\to C}$.
\end{proof}
\begin{claim}
For each $C^\ast \in \Cstarset$, each node $v\in C^\ast$ receives  $O(\sqrt{n}\cdot \deg_{C^\ast}(v))$ edges at Step 1.
\end{claim}
\begin{proof}
At Step 1, each node $v\in V_{C^\ast}$ receives a message from some $u\in S_{C\to C^\ast}$ only if $v\in \Sstar_{C^\ast\to C}$. More precisely, $v$ receives at most $|N_{S_{C\to C^\ast}}(u)|/\sqrt{n}\le \sqrt{n}$ edges in $E_C$ from this node $u \in V_C$. Thus the total number of edges received by $v$ is at most
\begin{align*}
\sqrt{n}\cdot\sum_{C\in \Cset \setminus \{C^\ast\},\:v\in \Sstar_{C^\ast\to C}}\deg_{S_{C\to C^\ast}}(v)&\le 
\sqrt{n} \cdot \left( O(\sqrt{n}) \cdot \max_{C\in \Cset \setminus \{C^\ast\},\:v\in \Sstar_{C^\ast\to C}}\{\deg_{S_{C\to C^\ast}}(v)\}\right)\\
&\le 
O(\sqrt{n}\cdot \deg_{C^\ast}(v)),
\end{align*}
where we use again the fact that $|\Cset| = O(\sqrt{n})$ and the definition of  $\Sstar_{C^\ast\to C}$.
\end{proof}

\begin{claim}\label{claim2}
For each $C^\ast \in \Cstarset$, the average degree of $C^\ast$ is at least $\max_{C}\{|S_{C\to C^\ast}|\}$, and thus $|E_{C^\ast}|/|V_{C^\ast}| = \Omega(|E'|/n)$. 
\end{claim}
\begin{proof}
By the definition of $S_{C\to C^\ast}$, each node $v\in S_{C\to C^\ast}$ has at least $\sqrt{n}$ incident edges crossing $V_C$ and $\Sstar_{C^\ast\to C}$.
By the definition of $\Sstar_{C^\ast\to C}$, we can associate $\sqrt{n}$  edges in $E_{C^\ast}$ incident to $\Sstar_{C^\ast\to C}$ for each edge crossing $V_C$ and $\Sstar_{C^\ast \to C}$.
Therefore, the number of edges in $E_{C^\ast}$ incident to $\Sstar_{C^\ast \to C}$ is at least
\[
\sqrt{n} \cdot \sqrt{n} \cdot |S_{C\to C^\ast}| = n \cdot |S_{C\to C^\ast}|.
\]
Since this is true for all clusters $C \in \Cset \setminus \{C^\ast\}$, we  conclude that the average degree in $C^\ast$ is at least $\max_{C}\{|S_{C\to C^\ast}|\}$.
\end{proof}

To apply Theorem~\ref{th:spars}, we still need to have $|E_{C^\ast}|/|V_{C^\ast}| = \Omega(\sqrt{n})$. Unfortunately, we are unable to guarantee this inequality. However, if we restrict ourselves to the nodes in $V_{C^\ast}$ that have incident edges in $\bar E$, we can show that their average degree is $\Omega(\sqrt{n})$, which is also enough for us to use Theorem~\ref{th:spars} (see the remark just after the statement of Theorem~\ref{th:spars}). Here is the precise statement that we need.

\begin{claim}
For each $C^\ast \in \Cstarset$, each node $v\in C^\ast$ with incident edges in $\bar{E}$ has $\deg_{C^\ast}(v) = \Omega(\sqrt{n})$.
\end{claim}
\begin{proof}
By the definition of $\bar{E}$, we have $v \in \Sstar_{C^\ast\to C}$ for some $C \in \Cset \setminus \{C^\ast\}$. By the definition of $\Sstar_{C^\ast\to C}$, we have $\deg_{C^\ast}(v) > \sqrt{n} \cdot \deg_{C}(v) \geq \sqrt{n}$, as required.
\end{proof}

\paragraph*{Summary.}  The algorithm guarantees that as long as $C \in \Cset$ is not a low-average cluster w.r.t.~$E'$ in Part~1 and Part~2, then all $K_4$ with at least one edge in $C$ are listed by some node.

 By Claim~\ref{claim:lowAvgC} and the fact that the number of inter-cluster edges in an expander decomposition is at most $|E|/\polylog(n)$, there are no more than $(1/\polylog(n)+  1/\gamma +1/\gamma')|E|$ remaining edges if the average degree $\avgdeg$ of the graph is at least $\avgdeg \geq n^{1/2}$ at the beginning. Choosing   $\gamma=8$, and $\gamma'=8$, gives that there are at most $3(1/8)|E| < |E|/2$ remaining edges.
In case the average degree~$\avgdeg$ of the graph is bounded by $\avgdeg < n^{1/2}$, the initial exhaustive search procedure for nodes of degree at most $2n^{1/2}$ removes at least half of the nodes by Claim~\ref{claim:lowDeg}.
This implies that we complete within a logarithmic number of iterations, which proves that within  $\tilde{O}(n^{1/2})$ rounds we list all instances of~$K_4$.

\section*{Acknowledgements} 
The authors would like to thank Yuval Efron and Miel Sharf for helpful discussions, and Orr Fischer for elaborating upon~\cite{Eden+DISC19}.

This project was partially supported by the European Union’s Horizon 2020 Research and Innovation Programme under grant agreement no.~755839.
YC was supported by Dr.~Max R\"{o}ssler, by the Walter Haefner Foundation, and by the ETH Z\"{u}rich Foundation.
FLG was supported by JSPS KAKENHI grants Nos.~JP16H01705, JP19H04066, JP20H00579, JP20H04139 and by the MEXT Quantum Leap Flagship Program (MEXT Q-LEAP) grant No.~JPMXS0120319794.

\bibliography{References}

\begin{thebibliography}{10}

\bibitem{Abboud+arxiv17}
Amir Abboud, Keren Censor{-}Hillel, Seri Khoury, and Christoph Lenzen.
\newblock Fooling views: {A} new lower bound technique for distributed
  computations under congestion.
\newblock {\em Distributed Computing}, 33:545--559, 2020.
\newblock URL: \url{https://doi.org/10.1007/s00446-020-00373-4}, \href
  {http://dx.doi.org/10.1007/s00446-020-00373-4}
  {\path{doi:10.1007/s00446-020-00373-4}}.

\bibitem{Alon+SIDMA08}
Noga Alon, Tali Kaufman, Michael Krivelevich, and Dana Ron.
\newblock Testing triangle-freeness in general graphs.
\newblock {\em SIAM Journal on Discrete Mathematics}, 22(2):786--819, 2008.
\newblock \href {http://dx.doi.org/10.1137/07067917X}
  {\path{doi:10.1137/07067917X}}.

\bibitem{becchetti+KDD08}
Luca Becchetti, Paolo Boldi, Carlos Castillo, and Aristides Gionis.
\newblock Efficient semi-streaming algorithms for local triangle counting in
  massive graphs.
\newblock In {\em Proceedings of the 14th ACM SIGKDD International Conference
  on Knowledge Discovery and Data Mining (KDD 2008)}, pages 16--24, 2008.
\newblock \href {http://dx.doi.org/10.1145/1839490.1839494}
  {\path{doi:10.1145/1839490.1839494}}.

\bibitem{Censor+PODC20}
Keren Censor{-}Hillel, Fran{\c{c}}ois~Le Gall, and Dean Leitersdorf.
\newblock On distributed listing of cliques.
\newblock In {\em Proceedings of the {ACM} Symposium on Principles of
  Distributed Computing ({PODC} 2020)}, 2020.
\newblock \href {http://dx.doi.org/10.1145/3382734.3405742}
  {\path{doi:10.1145/3382734.3405742}}.

\bibitem{Censor-Hillel+DC19}
Keren Censor{-}Hillel, Petteri Kaski, Janne~H. Korhonen, Christoph Lenzen, Ami
  Paz, and Jukka Suomela.
\newblock Algebraic methods in the congested clique.
\newblock {\em Distributed Computing}, 32(6):461--478, 2019.
\newblock URL: \url{https://doi.org/10.1007/s00446-016-0270-2}, \href
  {http://dx.doi.org/10.1007/s00446-016-0270-2}
  {\path{doi:10.1007/s00446-016-0270-2}}.

\bibitem{Chang+SODA19}
Yi-Jun Chang, Seth Pettie, and Hengjie Zhang.
\newblock Distributed triangle detection via expander decomposition.
\newblock In {\em Proceedings of the 30th Annual ACM-SIAM Symposium on Discrete
  Algorithms (SODA 2019)}, pages 821--840, 2019.
\newblock \href {http://dx.doi.org/10.5555/3310435.3310486}
  {\path{doi:10.5555/3310435.3310486}}.

\bibitem{Chang+PODC19}
Yi{-}Jun Chang and Thatchaphol Saranurak.
\newblock Improved distributed expander decomposition and nearly optimal
  triangle enumeration.
\newblock In {\em Proceedings of the {ACM} Symposium on Principles of
  Distributed Computing ({PODC} 2019)}, pages 66--73, 2019.
\newblock URL: \url{https://doi.org/10.1145/3293611.3331618}, \href
  {http://dx.doi.org/10.1145/3293611.3331618}
  {\path{doi:10.1145/3293611.3331618}}.

\bibitem{Czumaj+DISC18}
Artur Czumaj and Christian Konrad.
\newblock Detecting cliques in {CONGEST} networks.
\newblock In {\em Proceedings of the 32nd International Symposium on
  Distributed Computing ({DISC} 2018)}, pages 16:1--16:15, 2018.
\newblock URL: \url{https://doi.org/10.4230/LIPIcs.DISC.2018.16}, \href
  {http://dx.doi.org/10.4230/LIPIcs.DISC.2018.16}
  {\path{doi:10.4230/LIPIcs.DISC.2018.16}}.

\bibitem{Drucker+PODC14}
Andrew Drucker, Fabian Kuhn, and Rotem Oshman.
\newblock On the power of the congested clique model.
\newblock In {\em Proceedings of the {ACM} Symposium on Principles of
  Distributed Computing ({PODC} 2014)}, pages 367--376, 2014.
\newblock URL: \url{https://doi.org/10.1145/2611462.2611493}, \href
  {http://dx.doi.org/10.1145/2611462.2611493}
  {\path{doi:10.1145/2611462.2611493}}.

\bibitem{Eden+DISC19}
Talya Eden, Nimrod Fiat, Orr Fischer, Fabian Kuhn, and Rotem Oshman.
\newblock {Sublinear-time distributed algorithms for detecting small cliques
  and even cycles}.
\newblock In {\em Proceedings of the 33rd International Symposium on
  Distributed Computing (DISC 2019)}, pages 15:1--15:16, 2019.
\newblock URL: \url{https://doi.org/10.4230/LIPIcs.DISC.2019.15}, \href
  {http://dx.doi.org/10.4230/LIPIcs.DISC.2019.15}
  {\path{doi:10.4230/LIPIcs.DISC.2019.15}}.

\bibitem{Eden+SICOMP17}
Talya Eden, Amit Levi, Dana Ron, and C.~Seshadhri.
\newblock Approximately counting triangles in sublinear time.
\newblock {\em SIAM Journal on Computing}, 46(5):1603--1646, 2017.
\newblock \href {http://dx.doi.org/10.1137/15M1054389}
  {\path{doi:10.1137/15M1054389}}.

\bibitem{Even+DISC17}
Guy Even, Orr Fischer, Pierre Fraigniaud, Tzlil Gonen, Reut Levi, Moti Medina,
  Pedro Montealegre, Dennis Olivetti, Rotem Oshman, Ivan Rapaport, and Ioan
  Todinca.
\newblock Three notes on distributed property testing.
\newblock In {\em Proceedings of the 31st International Symposium on
  Distributed Computing ({DISC} 2017)}, pages 15:1--15:30, 2017.
\newblock URL: \url{https://doi.org/10.4230/LIPIcs.DISC.2017.15}, \href
  {http://dx.doi.org/10.4230/LIPIcs.DISC.2017.15}
  {\path{doi:10.4230/LIPIcs.DISC.2017.15}}.

\bibitem{Fischer+SPAA18}
Orr Fischer, Tzlil Gonen, Fabian Kuhn, and Rotem Oshman.
\newblock Possibilities and impossibilities for distributed subgraph detection.
\newblock In {\em Proceedings of the 30th Symposium on Parallelism in
  Algorithms and Architectures ({SPAA} 2018)}, pages 153--162, 2018.
\newblock URL: \url{https://doi.org/10.1145/3210377.3210401}, \href
  {http://dx.doi.org/10.1145/3210377.3210401}
  {\path{doi:10.1145/3210377.3210401}}.

\bibitem{Ghaffari+PODC17}
Mohsen Ghaffari, Fabian Kuhn, and Hsin-Hao Su.
\newblock Distributed {MST} and routing in almost mixing time.
\newblock In {\em Proceedings of the ACM Symposium on Principles of Distributed
  Computing (PODC 2017)}, pages 131--140, 2017.
\newblock URL: \url{https://doi.org/10.1145/3087801.3087827}, \href
  {http://dx.doi.org/10.1145/3087801.3087827}
  {\path{doi:10.1145/3087801.3087827}}.

\bibitem{Ghaffari+DISC18}
Mohsen Ghaffari and Jason Li.
\newblock {New distributed algorithms in almost mixing time via transformations
  from parallel algorithms}.
\newblock In {\em Proceedings of the 32nd International Symposium on
  Distributed Computing (DISC 2018)}, pages 31:1--31:16, 2018.
\newblock URL: \url{https://doi.org/10.4230/LIPIcs.DISC.2018.31}, \href
  {http://dx.doi.org/10.4230/LIPIcs.DISC.2018.31}
  {\path{doi:10.4230/LIPIcs.DISC.2018.31}}.

\bibitem{Gonen+OPODIS17}
Tzlil Gonen and Rotem Oshman.
\newblock Lower bounds for subgraph detection in the {CONGEST} model.
\newblock In {\em Proceedings of the 21st International Conference on
  Principles of Distributed Systems ({OPODIS} 2017)}, pages 6:1--6:16, 2017.
\newblock URL: \url{https://doi.org/10.4230/LIPIcs.OPODIS.2017.6}, \href
  {http://dx.doi.org/10.4230/LIPIcs.OPODIS.2017.6}
  {\path{doi:10.4230/LIPIcs.OPODIS.2017.6}}.

\bibitem{hu+JCSS16}
Xiaocheng Hu, Miao Qiao, and Yufei Tao.
\newblock {I/O}-efficient join dependency testing, {Loomis}--{Whitney} join,
  and triangle enumeration.
\newblock {\em Journal of Computer and System Sciences}, 82(8):1300--1315,
  2016.
\newblock \href {http://dx.doi.org/10.1016/j.jcss.2016.05.005}
  {\path{doi:10.1016/j.jcss.2016.05.005}}.

\bibitem{Huang+SODA20}
Dawei Huang, Seth Pettie, Yixiang Zhang, and Zhijun Zhang.
\newblock The communication complexity of set intersection and multiple
  equality testing.
\newblock In {\em Proceedings of the Fourteenth Annual ACM-SIAM Symposium on
  Discrete Algorithms}, pages 1715--1732, 2020.
\newblock \href {http://dx.doi.org/10.1137/1.9781611975994.105}
  {\path{doi:10.1137/1.9781611975994.105}}.

\bibitem{Izumi+PODC17}
Taisuke Izumi and Fran{\c{c}}ois~Le Gall.
\newblock Triangle finding and listing in {CONGEST} networks.
\newblock In {\em Proceedings of the {ACM} Symposium on Principles of
  Distributed Computing ({PODC} 2017)}, pages 381--389, 2017.
\newblock URL: \url{https://doi.org/10.1145/3087801.3087811}, \href
  {http://dx.doi.org/10.1145/3087801.3087811}
  {\path{doi:10.1145/3087801.3087811}}.

\bibitem{Izumi+STACS20}
Taisuke Izumi, Fran{\c{c}}ois {Le Gall}, and Fr{\'{e}}d{\'{e}}ric Magniez.
\newblock Quantum distributed algorithm for triangle finding in the {CONGEST}
  model.
\newblock In {\em Proceedings of the 37th International Symposium on
  Theoretical Aspects of Computer Science (STACS 2020)}, pages 23:1--23:13,
  2020.
\newblock URL: \url{https://doi.org/10.4230/LIPIcs.STACS.2019.49}, \href
  {http://dx.doi.org/10.4230/LIPIcs.STACS.2019.49}
  {\path{doi:10.4230/LIPIcs.STACS.2019.49}}.

\bibitem{JerrumSICOMP89}
Mark Jerrum and Alistair Sinclair.
\newblock Approximating the permanent.
\newblock {\em SIAM Journal on Computing}, 18(6):1149--1178, 1989.
\newblock \href {http://dx.doi.org/10.1137/0218077}
  {\path{doi:10.1137/0218077}}.

\bibitem{Korhonen+OPODIS17}
Janne~H. Korhonen and Joel Rybicki.
\newblock Deterministic subgraph detection in broadcast {CONGEST}.
\newblock In {\em Proceedings of the 21st International Conference on
  Principles of Distributed Systems ({OPODIS} 2017)}, pages 4:1--4:16, 2017.
\newblock URL: \url{https://doi.org/10.4230/LIPIcs.OPODIS.2017.4}, \href
  {http://dx.doi.org/10.4230/LIPIcs.OPODIS.2017.4}
  {\path{doi:10.4230/LIPIcs.OPODIS.2017.4}}.

\bibitem{LeGall+FOCS14}
Fran{\c{c}}ois Le~Gall.
\newblock Improved quantum algorithm for triangle finding via combinatorial
  arguments.
\newblock In {\em Proceedings of the 55th Annual IEEE Symposium on Foundations
  of Computer Science (FOCS 2014)}, pages 216--225, 2014.
\newblock \href {http://dx.doi.org/10.1109/FOCS.2014.31}
  {\path{doi:10.1109/FOCS.2014.31}}.

\bibitem{Pandurangan+SPAA18}
Gopal Pandurangan, Peter Robinson, and Michele Scquizzato.
\newblock On the distributed complexity of large-scale graph computations.
\newblock In {\em Proceedings of the 30th on Symposium on Parallelism in
  Algorithms and Architectures (SPAA 2018)}, pages 405--414, 2018.
\newblock \href {http://dx.doi.org/10.1145/3210377.3210409}
  {\path{doi:10.1145/3210377.3210409}}.

\bibitem{Shun+IDCE15}
Julian Shun and Kanat Tangwongsan.
\newblock Multicore triangle computations without tuning.
\newblock In {\em Proceedings of the 31st IEEE International Conference on Data
  Engineering (ICAD 2015)}, pages 149--160, 2015.
\newblock \href {http://dx.doi.org/10.1109/ICDE.2015.7113280}
  {\path{doi:10.1109/ICDE.2015.7113280}}.

\end{thebibliography}

\end{document}